\documentclass{article}

\newif\ifshort
\shortfalse

\usepackage{graphicx} 
\usepackage{amsmath}
\usepackage{amssymb}
\usepackage{amsfonts}
\usepackage{amsthm}
\usepackage{yfonts}
\usepackage[portrait, left=1in,right=1in,top=1in,bottom=1in,letterpaper]{geometry}
\usepackage[noadjust]{cite}
\usepackage{comment,enumerate,hyperref}
\usepackage[noend, boxed]{algorithm2e}
\usepackage{makecell}
\usepackage{thm-restate}

\makeatletter
\@ifundefined{newcounteralias}{}{%
  \renewcommand\thmt@autorefsetup{%
    \@xa\def\csname\thmt@envname autorefname\@xa\endcsname\@xa{\thmt@thmname}%
  }%
}
\makeatother

\declaretheorem[name=Theorem]{theorem}
\declaretheorem[name=Corollary,sibling=theorem]{corollary}
\declaretheorem[name=Lemma,sibling=theorem]{lemma}

\declaretheorem[name=Definition,sibling=theorem,style=definition]{definition}

\declaretheorem[name=Question,sibling=theorem,style=definition]{question}
\declaretheorem[name=Theorem,numbered=no,style=definition]{theorem*}

\newcommand{\dist}{\texttt{dist}}
\newcommand{\eps}{\varepsilon}
\newcommand{\mst}{\texttt{mst}}
\newcommand{\tee}{\mathcal{T}}

\usepackage[colorinlistoftodos,textsize=small,color=red!25!white,obeyFinal]{todonotes}

\newcommand{\codexnote}[2]{\todo[inline, size=\small, color=orange!20]{\textbf{Codex review #1:} #2}}

\title{Light Edge Fault Tolerant Graph Spanners}
\ifshort 
\else
\author{Greg Bodwin\thanks{University of Michigan. 
 Supported by NSF:AF 2153680.} \and Michael Dinitz\thanks{Johns Hopkins University.  Supported in part by NSF award 2228995.} \and Ama Koranteng\footnotemark[2] \and Lily Wang\footnotemark[1]}
\fi
\date{}

\begin{document}
\allowdisplaybreaks
\ifshort \else
\begin{titlepage}
\fi
\maketitle

\begin{abstract}
    Traditionally, the study of graph spanners has focused on the tradeoff between stretch and two different notions of size: the sparsity (number of edges) and the lightness (total weight, normalized by the weight of an MST).  There has recently been significant interest in \emph{fault tolerant} spanners, which are spanners that still maintain their stretch guarantees after some nodes or edges fail.  This work has culminated in an almost complete understanding of the three-way tradeoff between stretch, sparsity, and number of faults tolerated.  However, despite some progress in specific metric spaces (e.g., [Le-Solomon-Than FOCS '23]) there have been no results to date on the tradeoff in general graphs between stretch, \emph{lightness}, and number of faults tolerated.

    We initiate the study of light edge fault tolerant (EFT) graph spanners, obtaining the first such results:
    \begin{itemize}
    \item First, we observe that lightness can be unbounded in the fault-tolerant setting if we use the traditional definition (normalizing by the MST), even when tolerating just a single edge fault.
    
    \item We then argue that a natural definition of lightness in the fault-tolerant setting is to instead normalize by a min-weight fault tolerant connectivity preserver; essentially, a fault-tolerant version of the MST.
    However, even with this, we show a new lower bound establishing that it is still not generally possible to construct $f$-EFT spanners whose weight compares reasonably to the weight of a min-weight $f$-EFT connectivity preserver.
    
    \item In light of this lower bound, it is natural to then consider \emph{bicriteria} notions of lightness, where we compare the weight of an $f$-EFT spanner to a min-weight $(f' > f)$-EFT connectivity preserver.
    The most interesting question is to determine the minimum value of $f'$ that allows for reasonable lightness upper bounds.
    Our main result is a precise answer to this question: $f' = 2f$.
    More formally, we show that the bicriteria lightness can be untenably large (roughly $n/k$, for a $k$-spanner) if one normalizes by the min-weight $(2f-1)$-EFT connectivity preserver, but that it is bounded by just $O(f^{1/2})$ times the corresponding bound on non-fault tolerant lightness (roughly $n^{1/k}$, for a $(1+\eps)(2k-1)$-spanner) if one normalizes by the min-weight $2f$-EFT connectivity preserver instead.
    \end{itemize}

    Additional results include lower bounds on the $2f$-bicriteria lightness, improved $f$-dependence for $(2+\eta)f$-bicriteria lightness (for arbitrary constant $\eta > 0$), and a way to trade a slightly worse $f$-dependence for the ability to construct these spanners in polynomial time.
\end{abstract}
\ifshort \else
\thispagestyle{empty}
\end{titlepage}
\fi

\section{Introduction}

We study \emph{graph spanners}, a basic kind of graph sparsifier that preserves distances up to a small stretch factor:

\begin{definition} [Spanners]
Let $G = (V, E, w)$ be a weighted graph and let $k \ge 1$.
A \emph{$k$-spanner} of $G$ is an edge-subgraph $H = (V, E' \subseteq E, w)$ for which $\dist_H(u,v) \leq k \cdot \dist_G(u,v)$ for all $u,v \in V$.
\end{definition}

The value $k$ is called the \emph{stretch} of the spanner.  
Spanners were introduced by Peleg and Ullman~\cite{PU89sicomp} and Peleg and Sch\"{a}ffer~\cite{PS89} in the context of distributed computing.  Since then, they have found numerous applications, ranging from traditional distributed systems and networking settings~\cite{PU89jacm,awerbuch1990network,awerbuch1991cient,peleg2000distributed}, to efficient data structures like distance oracles~ \cite{TZ05}, to preconditioning of linear systems~\cite{elkin2008lower}, and many others.  

A  large fraction of all work on graph spanners focuses on the tradeoff between the stretch $k$ and the ``size'' of the spanner.
One way to measure size is by sparsity (total number of edges).  The tradeoff between stretch and sparsity was essentially resolved in a classic theorem by Alth\"ofer et al.~\cite{ADDJS93}:
\begin{theorem} [\cite{ADDJS93}] \label{thm:introsparsity}
For every positive integer $k \geq 1$, every $n$-node graph $G$ admits a $(2k-1)$-spanner $H$ with $|E(H)| \leq O(n^{1+1/k})$, and this bound is tight (assuming the Erd\H{o}s girth conjecture~\cite{girth}).
\end{theorem}

Besides having a spanner with few edges, in some applications one wants a spanner with small total edge weight.
However, nothing along the lines of Theorem \ref{thm:introsparsity} will be possible: that is, not all graphs admit $(2k-1)$-spanners of total weight $w(H) \le O(n^{1+1/k})$, or any other function of $n$.
This is because we can always scale the edge weights of $G$ as high as we like, and the edge weights of the spanner must scale accordingly.
So, in order to study existential results for low-weight spanners, we need to tweak the definition: the standard move is to study \emph{lightness}, which normalizes spanner weight by the weight of a minimum spanning tree ($\mst$).

\begin{definition} [Lightness]
Given a subgraph $H$ of a graph $G$, we define the lightness of $H$ (with respect to $G$) to be the quantity
\begin{align*}
    \ell(H \mid G) = \frac{w(H)}{w(\mst(G))}.
\end{align*}
We will also write $\ell(H) := \ell(H \mid H)$.
\end{definition}

This fixes the scale-invariance issue, and it is the dominant notion of ``weighted size'' in the study of spanners. 
It has been the subject of intensive study~\cite{ADDJS93, CDNS92, ENS14, CW16, LS23, BF25, Bodwin25}, and reasonably tight bounds are known:
\begin{theorem} [\cite{LS23, Bodwin25}] \label{thm:introlightness}
For every positive integer $k \geq 1$ and every $0<\eps\le1$, every $n$-node graph $G$ admits a $(1+\eps)(2k-1)$-spanner $H$ of lightness $\ell(H \mid G) \leq O(\eps^{-1} n^{1/k})$.
This dependence on $n$ is tight assuming the Erd\H{o}s girth conjecture~\cite{girth} (but the dependence on $\eps$ might be improvable).
\end{theorem}

\subsection{Fault Tolerance}

Another important aspect of spanners goes back to their origins in distributed systems and computer networking.  An issue that affects real-life systems that is not captured by the standard notion of spanners is the possibility of \emph{failure}.
If some edges (e.g., communication links) or vertices (e.g., computer processors) fail, what remains of the spanner might not still approximate the distances in what remains of the original graph.
This motivates the notion of \emph{fault tolerant} spanners, originally studied in the geometric setting by Levcopoulos, Narasimhan, and Smid \cite{LNS98} and then (more generally) the setting of doubling metrics; see, e.g.,~\cite{LNS98,Lukovszki99,CZ04,NS07,BHO20,Solomon14,BDMS13,CLN15,CLNS15,LST23} and references within.
In particular, the tradeoffs between the stretch and both the sparsity~\cite{LNS98,Lukovszki99,CZ04,CLN15,CLNS15,LST23} and the lightness~\cite{CZ04,Solomon14,CLNS15,LST23} for geometric spanners, including in the fault tolerant setting, have been the subject of significant study and interest, and optimal bounds are now known~\cite{CZ04,LST23}.  

The study of fault-tolerant spanners in general graphs was initiated in a seminal paper by Chechik, Langberg, Peleg, and Roditty~\cite{CLPR10}.  They introduced the following definition.

\begin{definition} [Edge Fault Tolerant Spanners \cite{CLPR10}] \label{def:FT}
A subgraph $H$ is an $f$-edge fault tolerant ($f$-EFT) $k$-spanner of $G = (V, E)$ if
$$\dist_{H \setminus F}(u,v) \leq k \cdot \dist_{G \setminus F}(u,v)$$
for all $u,v \in V$ and $F \subseteq E$ with $|F| \leq f$.
\end{definition}

Note that faults in general graphs are significantly more complex than in geometric settings, since in a general graph the failure of an edge or node can affect many distances (any pair that uses the failed object in the shortest path), while in a geometric setting the distances between nodes are fixed in the underlying metric space regardless of failures.  Nevertheless, Chechik et al.~\cite{CLPR10} were able to obtain the following strong upper bound.  

\begin{theorem} [\cite{CLPR10}]
For all positive integers $n, k, f$, every $n$-node weighted graph has an $f$-EFT $(2k-1)$-spanner $H$ on $|E(H)| \le O(f \cdot n^{1+1/k})$ edges. 
\end{theorem}

They also showed a bound of $|E(H)| \le \exp(f) \cdot n^{1+1/k}$ for an analogous notion of $f$-vertex fault tolerant (VFT) spanners (see Section \ref{sec:vft}).
Crucially, the dependence on $n$ in these theorems is identical to the dependence on $n$ in the non-faulty setting ($n^{1+1/k}$, from Theorem \ref{thm:introsparsity}).  Hence further work on sparsity of fault-tolerant spanners has focused on improving and lower bounding the dependence on the fault tolerance parameter $f$ in the spanner size.
This has led to a long line of work~\cite{CLPR10, DK11podc,BDPV18,BP19,DR20,BDR21,BDR22,Parter22, BHP24, PST24}, which gradually improved on the pioneering $f$-dependencies of \cite{CLPR10} in both settings.
This has culminated in optimal sparsity bounds for VFT spanners~\cite{BP19,BDR21}, near-optimal sparsity bounds for EFT spanners~\cite{BDR22}, and efficient construction algorithms \cite{DR20, BDR21, Parter22}.

But what about the \emph{lightness} of fault-tolerant spanners in the general graph setting?
An analogous upper bound would be that every $n$-node graph admits an $f$-EFT $(1+\epsilon)(2k-1)$-spanner $H$ of lightness roughly\footnote{Throughout the paper, we use $O_x$ notation to hide factors that depend on the variable(s) $x$, i.e., they are constant when $x$ is a constant.}
\begin{align} \label{eq:desired}
    \ell(H \mid G) \le O_{\eps, f}\left(n^{1/k}\right),
\end{align}
that is, matching the optimal dependence on $n$ from Theorem \ref{thm:introlightness}.
But to date, no such result has been achieved, despite the intensive study of both (non-faulty) spanner lightness and of fault-tolerant spanner sparsity. 
It is thus the next natural question for the area; indeed, the existence of light fault-tolerant graph spanners was explicitly raised as an open problem in a recent survey talk at the Simons Institute~\cite{HungLeTalk}.

In this paper we initiate the study of light fault-tolerant spanners.  
Our contribution is threefold.
First, we explain the lack of previous results by showing that the first two natural attempts to \emph{define} lightness in the fault-tolerant setting both encounter strong lower bounds, showing that no upper bound of this form (i.e., the form of~\eqref{eq:desired}) are possible.
Second, we propose a bicriteria notion of lightness (which we call \emph{competitive lightness}), and we pin down the exact threshold of bicriteria approximation at which analogous upper bounds for fault tolerant spanners become available.  Third, we provide bounds on the required and achievable dependence on $f$ in competitive lightness.

\subsection{Main Result: Defining Fault-Tolerant Lightness}

\subsubsection{(Failed) Attempt 1: Normalize by an MST}
The obvious first attempt to study light fault-tolerant spanners is to keep the same definition as before: lightness is the weight of the spanner divided by the weight of an MST.
However, it turns out that this definition suffers from a scaling issue, similar to the one discussed in the original definition of lightness: this measure of lightness can be unbounded.

More specifically, consider a graph $G$ \ifshort \else (see Figure~\ref{fig:badtriangle}) \fi on three nodes $u,v,w$, with $w(u,v) = w(u,w) = 1$ and $w(v,w) = W$.  Then $\mst(G)$ has two edges $\{u,v\}, \{u,w\}$ and total weight $2$, while any $1$-EFT spanner $H$ with finite stretch must include all three edges, and so will have weight $W+2$.
So we have
$$\frac{w(H)}{\mst(G)} \ge \Omega(W),$$
where $W$ may be selected as large as we like, demonstrating that this notion of lightness is unbounded.

\ifshort
\else
\begin{figure}[t]
\begin{center}
\begin{tikzpicture}
    \draw (0,0) -- (4,0) node[midway, below] {1};
    \draw (4,0) -- (2,3) node[midway, right] {1};
    \draw (2,3) -- (0,0) node[midway, left] {W};

    \fill (0,0) circle (2pt);
    \fill (4,0) circle (2pt);
    \fill (2,3) circle (2pt);
\end{tikzpicture}
\end{center}
\caption{\label{fig:badtriangle} This graph $G$ shows that the quantity $\frac{w(H)}{w(\mst(G))}$ can be unbounded for $1$-EFT spanners.
Here $w(G) = W+2$ for any $W$ we like, and $w(\mst(G))=2$, but we cannot remove any edges in a $1$-EFT spanner.  Thus a different definition of lightness is needed in the fault tolerant setting.}
\end{figure}
\fi

\subsubsection{(Failed) Attempt 2: Normalize by a Fault-Tolerant Connectivity Preserver}

Normalizing by the MST is not clearly the natural notion of fault tolerant lightness in the first place.
In the non-faulty setting, the MST is the cheapest subgraph that preserves \emph{connectivity}, and the lowest-weight spanner is the cheapest subgraph that preserves \emph{approximate distances}.
Thus non-faulty lightness can be interpreted as the ``relative price of approximating distances,'' as compared to just connectivity.

So arguably a more natural definition of fault tolerant lightness would measure the price of \emph{fault-tolerant distance approximation} vs.\ \emph{fault-tolerant connectivity}.
That is, instead of normalizing by an MST, our next attempt will be to normalize by 
the minimum weight \emph{$f$-edge fault connectivity preserver}~\cite{NI92,NI08, DKK22,DKKN23}\footnote{An $f$-EFT connectivity preserver is also sometimes called an $f+1$-connectivity certificate~\cite{NI92,NI08}.  This other terminology is more commonly used when the focus is on sparsity or computation time, rather than weights.}.
\begin{definition} [EFT Connectivity Preservers]
A subgraph $H \subseteq G$ is an \emph{$f$-edge fault tolerant (EFT) connectivity preserver}
if for any fault set $F \subseteq E$ with $|F| \leq f$, the connected components of $H \setminus F$ are identical to the connected components of $G \setminus F$.
Or equivalently, for all pairs of nodes $u, v \in V$, they are connected in $H \setminus F$ if and only if they are connected in $G \setminus F$.  
\end{definition}

This leads to a natural definition of lightness:

\begin{definition} [Competitive Lightness] \label{def:complight}
    Given a graph $G$, let $T_f(G)$ denote the set of all $f$-EFT connectivity preservers of $G$.  Then we define the $f$-competitive lightness of a subgraph $H$ of $G$ to be 
    \begin{align*}
        \ell_f(H \mid G) &:= \frac{w(H)}{\min_{Q \in {T_f(G)}} w(Q)}.
    \end{align*}
    We will also write $\ell_f(H) := \ell_f(H \mid H)$.
\end{definition}

We refer to $f$ in this definition as the ``lightness competition parameter'', or just the ``competition parameter.''
We note that $\ell_0(H \mid G) = \ell(H \mid G)$,\footnote{In particular, if $f = 0$ then $T_f(G)$ consists of all connected subgraphs, and so $\min_{Q \in T_f(G)} w(Q) = w(\mst(G))$.} and so we recover classical lightness as a special case.
  
This definition successfully escapes the previous lower bound
by setting $f > 0$, and in particular it is natural to hypothesize that all $n$-node graphs will admit an $f$-EFT $(1+\eps)(2k-1)$-spanner $H$ with competitive lightness $\ell_f(H \mid G) \le O_{\eps, f}(n^{1/k})$.
Unfortunately, we refute this possibility with a new lower bound: the dependence on $n$ for this notion of lightness must be essentially \emph{linear}.

\begin{restatable}{theorem}{bicriterialower} \label{thm:lower-bound-main}
    For any $f, k \geq 1$, there is a family of $n$-node weighted graphs $G$ for which every $f$-EFT $k$-spanner $H$ has competitive lightness
    $$\ell_{f}(H \mid G) \geq \Omega\left( \frac{n}{f^2 k}\right).$$
\end{restatable}

(See \ifshort\else Section \ref{sec:lowerbound} \fi for the proof.)
So, in light of this lower bound, we need to again revisit our search for a definition of fault-tolerant lightness that admits reasonable existential upper bounds.

\subsubsection{(Successful) Attempt 3: Bicriteria Competitive Lightness}

The fix we propose is to use the paradigm of \emph{bicriteria} approximation.
In our context, this means that we will escape the lower bound in Theorem \ref{thm:lower-bound-main} by comparing the quality of an $f$-EFT spanner to an $f'$-EFT connectivity preserver, for some $f' > f$.  In other words, we will consider the $f'$-competitive lightness of the best $f$-EFT spanner.

How high do we need to push $f'$ in order to enable the study of existential bounds for light EFT spanners?
The main result of this paper is an exact answer to this question:

\begin{theorem} [Main Result] \label{thm:intromain}
    For all positive integers $f, k, n$ and all $\eps > 0$:
    \begin{itemize}
    \item \textbf{(Upper Bound)} Every $n$-node graph\footnote{This and all of our other upper bounds extend to multigraphs as well.} $G$ has an $f$-EFT $(1+\eps)(2k-1)$-spanner $H$ with $2f$-competitive lightness
    $$\ell_{\color{blue} 2f}(H \mid G) \le \texttt{poly}(f, \eps) \cdot {\color{red} n^{1/k}}.$$
    
    \item \textbf{(Lower Bound)} There are $n$-node graphs $G$ for which any $f$-EFT $(1+\eps)(2k-1)$-spanner $H$ has $(2f-1)$-competitive lightness
    $$\ell_{\color{blue} 2f-1}(H \mid G) \ge \texttt{poly}(f, k) \cdot {\color{red} n}.$$
    \end{itemize}
\end{theorem}

Thus, the right answer is exactly $f'=2f$.
From a technical standpoint, we remark that there is some precedent for this factor of $2$ gap: it arises for a similar technical reason as in the Nash-Williams tree packing theorem~\cite{NashWilliams}, a classic structural result for fault-tolerant graph connectivity.

\subsection{The Bicriteria Price of Fault Tolerance}

Theorem \ref{thm:intromain} is the point of this paper: it initiates the study of $2f$-competitive spanner lightness, and it shows that this is essentially the strictest definition of fault-tolerant lightness that one can hope for.  But now that we have settled the right setting of the lightness competition parameter, we are in a similar situation to sparsity bounds after~\cite{CLPR10}: what is the right dependence on $f$ in the $2f$-competitive lightness?  We know from Theorem~\ref{thm:intromain} that $\texttt{poly}(f)$ is an upper bound, but can we be more precise? 

We begin to answer this question, proving upper and lower polynomial bounds.  Both our upper and lower bounds work by reducing to the setting of non-faulty light spanners.
In turn, the bounds for non-faulty light spanners can be understood via the ``weighted girth'' framework of~\cite{ENS14}.
We overview this formally in Section~\ref{sec:weighted-girth}, but informally, we define the weighted girth of $G$ to be the minimum over all cycles $C$ of the total weight of the cycle divided by the max weight edge of the cycle.
Let $\lambda(n,k)$ be the maximum (classical) lightness of any $n$-node graph with weighted girth greater than $k$.
Elkin, Neiman, and Solomon \cite{ENS14} proved that every $n$-node graph has a $k$-spanner of lightness at most $\lambda(n, k+1)$, and that this bound is tight.
We prove the following bounds on the $2f$-competitive lightness in terms of $\lambda$, which we can then instantiate with the known upper and lower bounds on this function:

\begin{theorem} \label{thm:introsmallcomp}
For all positive integers $f, k, n$ and all $\eps > 0$:
\begin{itemize}
\item \textbf{(Upper Bound)} Every $n$-node graph $G$ has an $f$-EFT $(1+O(\eps))(2k-1)$-spanner $H$ with competitive lightness
     \begin{align*}
     \ell_{\color{blue} 2f}(H \mid G) &\le f^{1/2} \cdot O\left(\lambda\left( n, (1+\eps)2k\right)\right)\\
                         &\le {\color{red} f^{1/2}} \cdot O_{\eps}\left(n^{1/k}\right) \tag{by $\lambda$ upper bounds from \cite{LS23, Bodwin25}}.
     \end{align*}

\item \textbf{(Lower Bound)} There are $n$-node graphs $G$ for which any $f$-EFT $(1+\eps)(2k-1)$-spanner $H$ has competitive lightness
     \begin{align*}
     \ell_{\color{blue} 2f}(H \mid G) &\ge \Omega\left(\lambda\left(\frac{n}{2f}, (1+\eps)(2k-1)+1\right)\right)\\
                         &\ge {\color{red} f^{-\frac{1}{k}}} \cdot \Omega\left(n^{1/k}\right) \tag{if $0<\eps\le 1/k$, assuming \cite{girth}.}
     \end{align*}
\end{itemize}
\end{theorem}

Note that it remains open \emph{whether $2f$-competitive lightness should correlate positively or negatively with $f$}.
This is an important gap, which we hope can be closed by future work.

We provide results for one more setting.  We know from Theorem~\ref{thm:intromain} that $2f$ is the smallest possible competitive parameter, and Theorem~\ref{thm:introsmallcomp} provides bound on the $2f$-competitive lightness, but what if we increase the competition parameter slightly?  Do we get even better bounds on the competitive lightness?  We show that better bounds are indeed possible if we tolerate a higher competition parameter of $(2+\eta)f$.
\begin{theorem} \label{thm:introbigcomp}
For all positive integers $f, k, n$ and all $\eps > 0, \eta > 0$:
\begin{itemize}
\item \textbf{(Upper Bound)} Every $n$-node graph $G$ has an $f$-EFT $(1+O(\eps))(2k-1)$-spanner $H$ with competitive lightness
     \begin{align*}
     \ell_{\color{blue} (2 + \eta)f}(H \mid G) &\le { O_{\eta}(1)} \cdot \lambda\left( n, (1+\eps)2k\right)\\
                         &\le {\color{red} O(1)} \cdot O_{\eps, \eta}\left(n^{1/k}\right) \tag{by $\lambda$ upper bounds from \cite{LS23, Bodwin25}}.
     \end{align*}

\item \textbf{(Lower Bound)} There are $n$-node graphs $G$ for which any $f$-EFT $(1+\eps)(2k-1)$-spanner $H$ has competitive lightness
     \begin{align*}
     \ell_{\color{blue} (2+\eta)f}(H \mid G) &\ge \Omega\left(\lambda\left(\frac{n}{(2+\eta)f}, (1+\eps)(2k-1)+1\right)\right)\\
                         &\ge {\color{red} f^{-\frac{1}{k}}} \cdot \Omega\left(n^{1/k}\right) \tag{if $0<\eps\le 1/k$, assuming \cite{girth}.}
     \end{align*}
\end{itemize}
\end{theorem}

In other words: for slightly higher competition parameters we can completely remove the $f$-dependence from the upper bound, but our lower bound does not degrade at all.
While there are still $\texttt{poly}(f)$ gaps in this setting, between $O(1)$ and $\Omega(f^{-1/k})$, this time we can at least say that competitive lightness is nonincreasing with $f$ and that the gap disappears in the limit of large $k$.

\subsection{Efficient Construction Algorithms}

The spanners in our previous theorems are all achieved by a simple greedy algorithm, which is only a light variant on the greedy algorithm that is standard in prior work \cite{BDPV18, BP19} (see Algorithm \ref{alg:ftgreedy}).
That is:
\begin{enumerate}
\item Initialize the spanner $H$ as a min-weight $2f$-EFT (or $(2+\eta)f$-EFT) connectivity preserver.
\item For each remaining edge $(u, v)$ in the input graph order of increasing weight, add $(u, v)$ to the spanner $H$ iff there exists a set of edge failures $F$ under which $\dist_{H \setminus F}(u, v) > k \cdot w(u, v)$.
\end{enumerate}

This algorithm is simple, and it trivially produces a correct spanner, but unfortunately it is not efficient.
In fact, both steps encode NP-hard problems, and thus run in exponential time.

First: computing a min-weight connectivity preserver generalizes the \textsc{$f+1$-Edge Connected Spanning Subgraph} problem, which is NP-hard and which has received significant research attention~\cite{CT00,GGTW09,GG12}).
Thanks to previous work, it is straightforward to handle this issue: $2$-approximations for min-weight $f$-EFT connectivity preserver have recently been developed~\cite{DKK22,DKKN23}, building off of the seminal work on Survivable Network Design by Jain~\cite{Jain}.
Using these approximation algorithms in step 1 will affect the lightness bounds only by a constant factor.

Second: testing whether there exists a fault set forcing us to add $(u, v)$ to the spanner encodes the \textsc{Length-Bounded Cut} problem, which is also NP-hard \cite{BEHKKPSS10}.
In prior work on spanner \emph{sparsity}, there have been two strategies to address this: a simple one with a minor penalty to sparsity~\cite{DR20}, or a complex one with no penalty to sparsity~\cite{BDR21}.  Unfortunately, the simpler of these approaches does not work at all for lightness.
The more complex one works, but for technical reasons it loses an extra $f^{1/2}$ in the $2f$-competitive lightness setting (but is lossless in the $(2+\eta)f$-competitive setting).
See Section~\ref{sec:running-time-overview} and \ifshort\else Section~\ref{sec:polytime} \fi for more details.
This gives the following dependencies for \emph{efficiently constructable} light EFT spanners:

\begin{restatable}{theorem}{algpolytime} \label{thm:alg-polytime}
For all positive integers $f, k, n$ and all $\eps > 0, \eta > 0$, there is a randomized polynomial time algorithm that takes as input an $n$-node weighted graph $G$, and with high probability returns an $f$-EFT $(1+O(\eps))(2k-1)$-spanner $H$ of competitive lightness

\begin{minipage}{.45\linewidth}
\begin{align*}
\ell_{\color{blue} 2f}(H \mid G) &\le O\left( f \cdot \lambda(n, (1+\eps)2k)\right)\\
                                 &\le {\color{red} f} \cdot O_{\eps}\left( n^{1/k} \right)
\end{align*}
\end{minipage}%
\begin{minipage}{0.1\linewidth}
or
\end{minipage}%
\begin{minipage}{.45\linewidth}
\begin{align*}
\ell_{\color{blue} (2+\eta)f}(H \mid G) &\le O_{\eta}\left( \lambda(n, (1+\eps)2k)\right)\\
                                        &\le {\color{red} O(1)} \cdot O_{\eps, \eta}\left( n^{1/k} \right).
\end{align*}
\end{minipage}
\end{restatable}

\subsection{Paper Outline}

We begin in Section~\ref{sec:technical-overview} with a high-level overview of our approach.  We then get into technical details in Section~\ref{sec:upper-main}.  In order to clearly show the main ideas, we warm up with an analysis that gives weaker bounds in Section~\ref{sec:warmup}.
The proofs of our main upper bounds in Theorems~\ref{thm:intromain},~\ref{thm:introsmallcomp}, and \ref{thm:introbigcomp} are in Section~\ref{sec:multiple-host}.  \ifshort\else The proof of our polytime algorithm (Theorem~\ref{thm:alg-polytime}) is in Section~\ref{sec:polytime}.
The lower bound parts of these theorems appear in Section~\ref{sec:lowerbound}.\fi

\section{Technical Overview}  \label{sec:technical-overview}
The bulk of the technical work in this paper involves analyzing the (slightly modified) fault tolerant greedy algorithm to prove the upper bound part of our main results.
These proofs are closely connected.

\subsection{Overview of Upper Bounds}\label{overview}

A highly successful strategy to analyze the \emph{sparsity} of \emph{non-fault tolerant} spanners is to construct them using a simple greedy algorithm, and show that it produces spanners of high girth (shortest cycle length) \cite{ADDJS93}.
Then one can apply results from extremal graph theory stating that all high-girth graphs, including these spanners, must be sparse.
Previous work has independently extended this method to analyze the \emph{sparsity} of fault tolerant spanners \cite{BDPV18, BP19, BDR22}, and to the lightness of \emph{non-fault tolerant} spanners \cite{ENS14}.
Since we deal with light fault tolerant spanners, our proof strategy can broadly be described as a composition of these two extensions, as well as some new technical tools that handle the issues that arise from their interaction.

\paragraph{Algorithm and Blocking Sets.}

As discussed, our algorithm for constructing a $f$-EFT $k$-spanner is the following.  Let $Q$ be the optimal $2f$-EFT connectivity preserver of $G$, i.e., the graph that we are competing with (the weight in the denominator of our lightness definition).
We begin by adding $Q$ to the spanner $H$.  Then we use the standard greedy algorithm: we consider each other edge $e \in E(G) \setminus Q$ in nondecreasing weight order, and we add $e = \{u,v\}$ to $H$ if there is some set $F_e \subseteq E(H)$ with $|F_e| \leq f$ such that $\dist_{H \setminus F_e}(u,v) > k \cdot \dist_{G \setminus F_e}(u,v)$.  We say that $e$ is in a \emph{block} with each $e' \in F_e$, or equivalently that $(e, e')$ form a block for each such $e'$.  Note that $e$ is the first edge of a block at most $|F_e| \leq f$ times.  The collection of all blocks is called a blocking set.

\paragraph{Weighted Girth.}
Next, we explain the relationship to the weighted girth framework from \cite{ENS14}.
Standard arguments (first introduced by~\cite{BP19}) can be adapted to our modified algorithm to show that essentially every cycle in $H$ with at most $k+1$ edges must contain two edges that form a block, i.e., the blocks cover all cycles with at most $k+1$ edges.
Following a proof from \cite{ENS14}, this argument can be adapted to show that they also cover essentially all cycles with ``normalized weight'' at most $k+1$, i.e., cycles $C$ where $w(C) / \max_{e \in C} w(e)$ is at most $k+1$.
This is useful because if there were \emph{no} cycles of normalized weight at most $k+1$ --- or, in the language of \cite{ENS14}, $H$ has weighted girth $>k+1$ --- then by definition $H$ has lightness at most $\lambda(n,k+1)$.
So we have a spanner $H$ where cycles of normalized weight at most $k+1$ exist but are blocked, and we know that if there were \emph{no} such cycles then we have good lightness bounds.  So our goal will be to turn $H$ into some other graph $H'$ without \emph{any} cycles of normalized weight at most $k+1$, in a way that allows us to bound the lightness of $H$ using the lightness of $H'$.

\paragraph{The Subsampling Method.}

In the sparsity setting, one of the now-standard ways of going from a graph $H$ where short cycles exist but are blocked to a graph $H'$ where there are no short cycles is by subsampling $H$ to get $H'$. Informally,  by choosing the correct probability, one can often show that no blocks survive in $H'$, and thus there cannot be any short cycles and so there are not many edges in $H'$.  On the other hand, since we obtained $H'$ by subsampling $H$ with a known probability, there are not too many more edges in $H$ than in $H'$.  Thus $H$ cannot have many edges.  
A natural first attempt is to use the same idea for lightness, but it immediately runs into a serious problem.
The connectivity preserver $Q$ itself can have unblocked cycles of low normalized weight, and so if the subsampling does not remove any edges of $Q$ then we cannot hope to get rid of all such cycles.
However, if the subsampling \emph{does} get rid of edges of $Q$, then the min-weight connectivity preserver of the subsampled graph $H'$ might be very different from $Q$, making it hard to compare the lightnesses of $H$ and $H'$.
(Note that this problem does not show up when analyzing sparsity, since the number of edges is not normalized by anything and so it is relatively easy to compare $|E(H)|$ to $|E(H')|$).  

To get around this, we first observe that if $Q$ happened to be a tree that was disjoint from the blocking set then we would be in good shape:  none of the low-normalized-weight cycles would use $Q$, so we could do standard subsampling of $E(H) \setminus Q$ and then include $Q$ deterministically to get a subgraph with high weighted girth where $Q$ still exists. Unfortunately, $Q$ will not be a tree.  But suppose that we could find at least $f+1$ disjoint spanning trees in $Q$.  Then since each edge $e \in E(H) \setminus Q$ is blocked by at most $f$ other edges (the edges in $F_e$), we can find at least one of these trees which does not contain any edge blocked with $e$.  We can then add $e$ to this tree.  Once we do this for every $e \in E(H) \setminus Q$, we have a partition of $E(H) \setminus Q$ into $f+1$ disjoint subgraphs with the property that each subgraph contains a spanning tree of $Q$, and no edge in that subgraph includes a block with the spanning tree.  So we're in the setting we want to be in, and we can show that each of these subgraphs has good lightness.  Since $Q$ and $H$ are both partitioned across these subgraphs, this implies that we have good lightness overall.

\paragraph{Steiner Forest Packing.}  How can we find $f+1$ disjoint spanning trees in $Q$?
One very useful tool for finding disjoint spanning trees in graphs is the tree packing theorem of Nash-Williams~\cite{NashWilliams}, which implies that it suffices for $Q$ to be $2f+2$-connected.
This is close to what we want, except for two issues:
\begin{enumerate}
\item First, $Q$ is a $2f$-EFT connectivity preserver, which means that if $G$ is $2f+1$-connected (or more) then $Q$ is $2f+1$-connected.
We desire $2f+2$-connectivity, so we are off by one from our desired connectivity threshold.

\item Second, $G$ need not be $2f+1$-connected at all, and so $Q$ need not even be $2f+1$-connected.
\end{enumerate}
Fortunately, it turns out that we can get around both of these issues via a ``doubling trick'' and by applying bounds on \emph{Steiner Forest packing} rather than spanning tree packing.
Since $Q$ is a $2f$-EFT connectivity preserver, for any edge $(u, v) \in E(H) \setminus Q$ we know that the pair $(u, v)$ is $2f+1$-connected in $G$, or else we would have to include $(u, v)$ in $Q$.  So while $Q$ is not $2f+1$-connected, every pair $u,v$ that are endpoints of an edge in $E(H) \setminus Q$ are $2f+1$-connected in $G$ and thus in $Q$.  It turns out that this actually implies that there are $\Omega(f)$ disjoint \emph{forests} in $Q$ so that for each $\{u,v\} \in E(H) \setminus Q$, $u$ and $v$ are connected in all of the forests.  In other words, we can find $\Omega(f)$ Steiner forests when the demand pairs are $e \in E(H) \setminus Q$.  We would now be done if this $\Omega(f)$ was $f+1$, but unfortunately it is more like $f/16$ ~\cite{Lau05}; improving the constant in Steiner forest packing is a fascinating open problem, and the correct bound is not even known for Steiner trees~\cite{KRIESELLeven,DMP16}.

So we add one more step: we first double every edge of $Q$ to get a multigraph $Q'$.  Now every $\{u,v\} \in E(H) \setminus Q$ is $4f+2$-connected in $Q'$ and, importantly, all degrees are even and thus $Q'$ is Eulerian.  This turns out to make Steiner Forest packing significantly easier, and optimal bounds are known: it was shown by Chekuri and Shepherd~\cite{chekuri2009approximate} that if all demand pairs $u,v$ are $2k$-connected in an Eulerian graph $G$, then we can find $k$ disjoint Steiner forests in $G$ (and can even do so in polynomial time). This means that we can find $2f+1$ disjoint Steiner forests in $Q'$, and thus when interpreted in $Q$ we get $2f+1$ Steiner forests that are not disjoint, but where every edge of $Q$ is in at most $2$ of them. 

But this is sufficient for us.  Since each edge $e$ is in a block with at most $f$ edges of $Q$, and each edge of $Q$ can appear at most two of $2f+1$ Steiner forests, we get that there must be some Steiner forest which has no edge blocked with $e$, as desired.

To sum up: to bound the $2f$-competitive lightness of $H$, we double every edge of $Q$ and use the Steiner Forest packing in Eulerian graphs result of~\cite{chekuri2009approximate} to find $2f+1$ edge-disjoint Steiner forests, which correspond to $2f+1$ Steiner forests in $Q$ where every edge in $Q$ appears at most twice.  Then for every edge $e \in E(H) \setminus Q$ we add it to a Steiner forest which does not contain any of its blocks.  Finally, for each of the resulting subgraphs we use subsampling techniques to bound their lightness.  

This approach turns out to give a lightness bound of $f \cdot \lambda(n, k+1)$ (see Section~\ref{sec:warmup}).
In order to improve the bound to $f^{1/2} \cdot \lambda(n, k+1)$, we need to be a bit more careful.  At a very high level, we allow ourselves to add edges of $E(H) \setminus Q$ to \emph{many} of the Steiner forests at once, where possible.
If it is possible then we get an improved bound, roughly since our previous analysis overcounts multiply-added edges, letting us divide out an additional factor.
For edges that cannot be added to many Steiner forests, it turns out the previous subsampling technique can be made more efficient, also giving an improved bound in this other case.

The upper bound in Theorem~\ref{thm:introbigcomp}, with a higher competition parameter, has an almost identical proof to Theorem~\ref{thm:introsmallcomp}.  The main difference is that, since we now have that $Q$ is a $(2+\eta)f$-EFT connectivity preserver, if we repeat the above analysis then for any $e \in E(H) \setminus Q$ we can now find $\Omega_{\eta}(f)$ Steiner forests in the forest packing of $Q$ that don't contain a block with $e$, rather than just $1$.  Thus if we add the weights of all of the subgraphs, we are overcounting every edge by a factor of $\Omega(f)$.  When we plug this into the previous calculations, we end up exactly canceling out the $f$-dependence, giving the improved upper bound.

\subsection{Running Time (Theorem~\ref{thm:alg-polytime})} \label{sec:running-time-overview}
One downside of the fault tolerant greedy algorithm is that it does not run in polynomial time.  First, we need to seed it with $Q$, which as mentioned above requires solving an NP-hard problem.  This is straightforward to get around, though, since we can just use a $2$-approximation for the min-weight $f$-EFT connectivity preserver problem due to~\cite{DKK22,DKKN23}.  The larger difficulty is that we need to check for every edge $e \in E(H) \setminus Q$ whether there is a fault set which forces us to include $e$.  Doing this in the obvious way takes $\Omega(n^f)$ time.  This drawback has been noticed since the fault tolerant greedy algorithm was introduced~\cite{BDPV18}, and eventually two different methods were developed to design polytime versions of the algorithm:
\begin{enumerate}
    \item The first approach, introduced by~\cite{DR20} for sparse fault-tolerant spanners, involves giving an approximation algorithm for the question of whether there exists a fault set forcing us to add $e$.  This problem is NP-hard, but there are simple $O(k)$-approximations.  By carefully analyzing the effect of approximating rather than solving this problem, \cite{DR20} showed that this incurs only a relatively small extra loss in the sparsity.  This approach has since been used extensively to to turn variants of the fault tolerant greedy algorithm into polynomial-time algorithms, albeit with a small extra loss~\cite{BDR22,BDN22,BDN23,BHP24,PST24}.
    
    \item In order to avoid this extra loss,~\cite{BDR21} designed a very different subsampling-based condition to use in the greedy spanner.  Informally, the main idea of this technique is to move the subsampling from the analysis into the actual algorithm.  So instead of just subsampling to analyze sparsity, we actually subsample to decide whether to include an edge.     This method is conservative (it may add some edges even when there is no fault set forcing us to do so), but it can be proved that it does not add asymptotically more edges than the base fault tolerant greedy algorithm.  Perhaps due to its relative complexity, this approach has not proved to be nearly as useful or popular as the approximation algorithm-based approach of~\cite{DR20}.
\end{enumerate}

It turns out that making our version of fault-tolerant greedy polynomial time without losing its lightness properties is significantly more complex than if we cared only about sparsity.  Interestingly, the approximation algorithm approach of~\cite{DR20} does not work, as it fundamentally depends on being able to treat weighted edges as unweighted.  This turns out to be OK for sparsity, but makes it impossible to use for lightness.  However, the subsampling approach of~\cite{BDR21} \emph{can} be made to work.  Interestingly, this is the first setting (to the best of our knowledge) where the approximation algorithm approach does not work but the subsampling does; until now, it seemed as if the approximation approach was at least as flexible (and significantly easier to use) as the subsampling approach.

To get this approach to work we have to move the subsampling from the analysis to the algorithm.  It turns out that we can do this by actually computing the Steiner Forest packing of~\cite{chekuri2009approximate} (fortunately, they showed that it could be computed efficiently).  Then when we consider adding an edge $e$, for every Steiner forest in the packing we use the sampling idea from~\cite{BDR21} on the graph consisting of $T$ and the already-added spanner edges that are not in $Q$.  This basically suffices for Theorems~\ref{thm:alg-polytime}.  Unfortunately, the method we used before to improve the $f$ dependence to $f^{1/2}$ for $2f$-competitive lightness (in Theorem~\ref{thm:introsmallcomp}) required subsampling different types of edges with different probabilities (depending on whether or not an edge could be added to many of the Steiner forests).  Since the subsampling is now in the algorithm rather than the analysis, we cannot do this, and thus must fall back on a lightness bound of $f \cdot \lambda(n,k+1)$.

\section{Constructions of Light Fault-Tolerant Spanners} \label{sec:upper-main}

\subsection{The Weighted Girth Framework} \label{sec:weighted-girth}

Elkin, Neiman, and Solomon \cite{ENS14} introduced a weighted analog of girth, for use in studying the lightness of (non-faulty) spanners:

\begin{definition} [Normalized Weight and Weighted Girth \cite{ENS14}]
For a cycle $C$ in a weighted graph $G$, the \emph{normalized weight} of $C$ is the quantity
$$w^*(C) := \frac{w(C)}{\max_{e \in C} w(e)}.$$
The \emph{weighted girth} of $G$ is the minimum normalized weight over its cycles.
\end{definition}

They then proved an equivalence between light spanners and the maximum possible lightness of a graph of high weighted girth.
In particular:
\begin{definition} [Extremal Lightness of Weighted Girth]
We write $\lambda(n, k) := \sup \ell(G)$, 
where the $\sup$ is taken over $n$-node graphs $G$ of weighted girth $>k$.
\end{definition}

\begin{theorem} [\cite{ENS14}] \label{thm:ENS}
For all positive integers $n$ and all $k$, every $n$-node graph $G$ has a $k$-spanner $H$ of lightness
$\ell(H \mid G) \le \lambda(n, k+1)$,
and this is existentially tight.\footnote{More specifically, this means that for any $\eps > 0$, there exists an $n$-node graph on which any $k$-spanner $H$ has lightness $\ell(H \mid G) > \lambda(n, k+1) - \eps$.}
\end{theorem}

Settling the asymptotic value of $\lambda$ is a major open problem.
Currently, the following bounds are known:

\begin{theorem} [\cite{LS23, Bodwin25, BF25}] \label{thm:priorlight}
For all positive integers $n, k$ and all $0<\eps\le1$, we have
$$\lambda(n, (1+\eps)2k) \le O\left(\eps^{-1} n^{1/k} \right).$$
For constant $k\ge 2$ and $\eps = \Theta(n^{-\frac{1}{2k-1}})$, assuming the Erd\H{o}s girth conjecture \cite{girth} with parameter $k-1$, we have
$$\lambda(n, (1+\eps)2k) \ge \Omega\left(\eps^{-1/k} n^{1/k} \right).$$
\end{theorem}

\subsection{The Greedy Algorithm and Blocking Sets}

We will analyze spanners that arise from a variant of the fault-tolerant greedy algorithm, introduced in \cite{BDPV18} and used in many recent papers on fault tolerance.
The only difference algorithmically is that we seed the spanner with a min-weight connectivity preserver, whereas prior work immediately enters the main loop.

\begin{algorithm}
\DontPrintSemicolon

\textbf{Input:} Graph $G = (V, E, w)$, stretch $k$, fault tolerance $f$\;~\\

Let $Q \gets $ min-weight $2f$-FT connectivity preserver of $G$\;
Let $H \gets Q$ be the initial spanner\;
\ForEach{edge $(u, v) \in E(G) \setminus Q$ in order of nondecreasing weight}{
    \If{there exists $F \subseteq E(H), |F| \le f$ such that $\dist_{H \setminus F}(u, v) > k \cdot w(u, v)$}{
        add $(u, v)$ to $H$\;
    }
}
\textbf{Return} $H$\;
\caption{\label{alg:ftgreedy} Light Fault-Tolerant Greedy Spanner Algorithm}
\end{algorithm}

The proof of correctness of the algorithm is standard:
\begin{theorem} \label{thm:feasible}
The output spanner $H$ from Algorithm \ref{alg:ftgreedy} is an $f$-EFT $k$-spanner of the input graph $G$.
\end{theorem}
\begin{proof}
As is well known in the spanners literature (e.g.~\cite{ADDJS93}), it suffices to verify that for any set $F$ of $|F| \le f$ edge faults and any edge $(u, v) \in E(G) \setminus (E(H) \cup F)$, we have
$$\dist_{H \setminus F}(u, v) \le k \cdot w(u, v).$$
To show this, we first observe that for any such edge we have $(u, v) \notin Q$ (since we add all of $Q$ to $H$), and so we consider $(u, v)$ at some point in the main for-loop of the algorithm.
We choose not to add $(u, v)$ to $H$, meaning that for all possible fault sets $F$ we have $\dist_{H \setminus F}(u, v) \le k \cdot w(u, v)$.
For the rest of the algorithm we only add edges to $H$, and since this property is monotonic in $H$, it will still hold in the final spanner.
\end{proof}

Now we turn to analyzing lightness.
When $f=0$, this algorithm produces a spanner $H$ of weighted girth $>k+1$ \cite{ENS14}.
For larger $f$, there is not much we can say about the weighted girth of the output spanner, which could be quite small.
However, we might intuitively expect the output spanner to be ``structurally close'' to a graph of high weighted girth.
That notion of closeness can be formalized by adapting the blocking set framework, which has been used in many recent papers on fault-tolerant spanners and related objects (see, e.g., \cite{BP19, DR20, BDR21, BDR22,BDN22,BDN23, PST24, BHP24}).

\begin{definition} [Edge-Blocking Sets]
For a graph $H$, an \emph{edge-blocking set} is a set of ordered edge pairs $B \subseteq E(H) \times E(H)$.
We say that $B$ \emph{blocks} a cycle $C$ if there is a pair $(e_1, e_2) \in B$ with $e_1, e_2 \in C$.
We say that $B$ is \emph{$f$-capped} if for all $e \in E(H)$, there are at most $f$ edges $e'$ with $(e, e') \in B$.\footnote{There may be additional pairs of the form $(e', e) \in B$ -- this property only bounds the number of pairs that have $e$ first.}
\end{definition}

The following is a tweak on a standard lemma bounding the size of the blocking set of the output spanner; the proof is a straightforward mixture of related lemmas from \cite{BP19} and \cite{ENS14}.

\begin{lemma} \label{lem:bsetexists}
The output spanner $H$ from the fault-tolerant greedy algorithm has an $f$-capped edge-blocking set $B$ that blocks all cycles $C$ with normalized weight $w^*(C) \le k+1$ and where a heaviest edge in $C$ is not in $Q$.
Additionally, for every pair $(e, e') \in B$, the first edge $e$ is not in $Q$.
\end{lemma}
\begin{proof}
In the main part of the construction, whenever we add an edge $(u, v) \in E(G) \setminus Q$ to the spanner $H$, we do so because there exists a fault set $F_{(u, v)}$ that satisfies the conditional (there might be many possible choices of $F_{(u, v)}$, in which case we fix one arbitrarily).
Let
$$B := \left\{ (e, e') \ \mid \ e \in E(H) \setminus Q, e' \in F_e\right\}.$$
Since each $|F_e| \le f$, we have that $B$ is $f$-capped.
To argue its cycle-blocking properties, let $C$ be a cycle as in the lemma statement, and choose $(u,v)\notin Q$ to be the last-processed maximum-weight edge of $C$ outside $Q$.
When we consider $(u, v)$ in the main loop of the algorithm, just before $(u, v)$ is added to $H$, there exists a $u \leadsto v$ path through the other edges of the cycle, and so
\begin{align*}
\dist_H(u, v) &\le w(C) - w(u, v) \le (k+1) \cdot w(u, v) - w(u, v) = k \cdot w(u, v).
\end{align*}
However, in order to satisfy the conditional and include $(u, v)$ in $H$, we must have
$\dist_{H \setminus F_{(u, v)}}(u, v) > k \cdot w(u, v)$.
Thus at least one edge $e \in C \setminus \{(u, v)\}$ must be included in $F_{(u, v)}$.
So we have $((u, v), e) \in B$, which blocks $C$, completing the proof.
\end{proof}

A subtlety here is that the output spanner $H$ could still have unblocked cycles $C$ of low weighted girth, in the case where the heaviest edge of $C$ is in $Q$ (and so it was added before the main greedy loop of the algorithm), and the edges in $C \setminus Q$ are much lighter.
In general, min-weight $f$-EFT connectivity preservers are less well-behaved than min-weight spanning trees, and so situations like this are indeed possible even despite $Q$ being min-weight.
We will need to handle these cycles in our analysis in another way (see Lemma \ref{lem:highwg}).

\subsection{Warmup: Slow Suboptimal Competitive Lightness Upper Bounds} \label{sec:warmup}

As a warmup, we will prove:
\begin{theorem} [Warmup] \label{thm:mainbadf}
For all input graphs $G$, the output spanner $H$ from Algorithm \ref{alg:ftgreedy} satisfies
$$\ell_{2f}(H \mid G) \le O\left( f \cdot \lambda(n, k+1) \right).$$
\end{theorem}

We will later improve the dependence on $f$, and show a variant that runs in polynomial time.  But it is easier to demonstrate our main ideas with this simpler result.
We may assume without loss of generality that $f$ is at least a sufficiently large constant, which will help avoid corner cases in some sampling bounds below.

\paragraph{Extensions of Nash-Williams Tree Decompositions.}

To begin the proof, we recall the classic Nash-Williams tree packing theorem:
\begin{theorem} [Nash-Williams \cite{NashWilliams}] \label{thm:NW}
For all positive integers $f$, a multigraph $G$ contains $f$ edge disjoint spanning trees if and only if for every vertex partition $\mathcal{P}$, there are at least $f(|\mathcal{P}|-1)$ edges between parts.
\end{theorem}

The following is a simple corollary.  The actual packing that we will use is a bit more complicated, but we give this simpler corollary to build intuition. 

\begin{corollary} \label{cor:nwdecomp}
For all positive integers $f$, an $f$-connected graph $G$ contains a collection of $f$ spanning trees with the property that any edge of $G$ is in at most two of the trees.
\end{corollary}
\begin{proof}
Replace every edge of $G$ with two parallel edges, so that it is $2f$-connected.
Then consider any vertex partition $\mathcal{P}$, and note that (by connectivity) every part must have at least $2f$ outgoing edges.
Thus there are at least $f \cdot |\mathcal{P}|$ edges between parts.  
This is even a bit stronger than the condition needed to apply the Nash-Williams theorem, and so by Theorem~\ref{thm:NW} we can partition $G$ into $f$ edge-disjoint spanning trees.  Since we initially doubled the edges of $G$, this means that every original edge participates in at most two spanning trees.
\end{proof}

If the input graph $G$ is $(2f+1)$-connected (and therefore $Q$ is also $(2f+1)$-connected), then this corollary would give the right technical tool.
However, in order to avoid assuming any connectivity properties of the input graph, we instead need a generalization.
We begin with the following theorem by Chekuri and Shepherd, which gives the appropriate analog in the special case of Eulerian graphs:
\begin{theorem} [\cite{chekuri2009approximate}, c.f.\ Theorem 3.1]\label{lem:pairwisenw}
For all positive integers $f$ and Eulerian graphs $G$, there exist edge-disjoint forests $\{F_1, \dots, F_f\}$ such that every $2f$-connected component of $G$ is connected in all forests.
Moreover, these forests can be constructed in $\text{poly}(n)$ time.
\end{theorem}

We will use the following corollary:
\begin{corollary} \label{cor:pairwisenw}
For all positive integers $f$ and graphs $G$, there exists a collection of subtrees $\tee$ such that (1) every edge in $E(G)$ is in at most two trees in $\tee$, and (2) for every $f$-connected component $C$ of $G$, there are at least $f$ trees in $\tee$ in which $C$ is connected.  
\end{corollary}
\begin{proof}
Identical to the previous corollary.
Note that when we double edges, every node has even degree, so we guarantee that the graph is Eulerian and so can apply Theorem~\ref{lem:pairwisenw}.\footnote{Technically this edge-doubling step makes $G$ a multigraph, which is not explicitly mentioned in the Chekuri-Shepherd theorem \cite{chekuri2009approximate}.  However, their proof extends straightforwardly to multigraphs.}
\end{proof}

The first step in our analysis of Algorithm \ref{alg:ftgreedy} is to apply this corollary to $Q$, with parameter $2f+1$, yielding a collection of subtrees $\tee$.
We may assume without loss of generality that all subtrees in the collection have at least one edge, and hence have positive weight.

\paragraph{Construction of $H[T]$ Subgraphs.} Our next step is to construct a sequence of four subgraphs $H[T] \supseteq H'[T] \supseteq H''[T] \supseteq H'''[T]$ associated to each tree $T \in \tee$.
In the following let $B$ be an $f$-capped blocking set for $H$ as in Lemma \ref{lem:bsetexists}.

\begin{itemize}
\item \textbf{(Construction of $H[T]$)}
Consider the edges in $(u, v) \in E(H) \setminus Q$ one at a time.
Notice that there cannot exist $2f$ edge faults in $Q$ that disconnect the nodes $u, v$, since otherwise we would need to include the edge $(u, v)$ in $Q$.
Thus the node pair $(u, v)$ is $(2f+1)$-connected in $Q$.
So the endpoint nodes $u, v$ lie in the same $(2f+1)$-connected component $C$ of $Q$, and by Corollary \ref{cor:pairwisenw} there exist $2f+1$ trees in $\tee$ that all span $C$.

Since $B$ is $f$-capped, there are at most $f$ edges $e'$ with $(e, e') \in B$.
Since each such edge $e'$ can be in at most two trees (by Corollary~\ref{cor:pairwisenw}), it follows that there exists at least one tree $T \in \tee$ for which no such edge $e'$ is in $T$ (if there are multiple such trees, choose one arbitrarily).
We choose such a tree $T$ and say that it \emph{hosts} $e$.
Then, for all $T \in \tee$, let $H[T]$ be the graph that contains $T$ and all edges hosted by $T$.

\item \textbf{(Construction of $H'[T]$)}
For each tree $T$, we construct $H'[T]$ by keeping every edge in $T$ deterministically, and keeping every edge in $H[T] \setminus T$ independently with probability $1/f$.

\item \textbf{(Construction of $H''[T]$)}
To construct $H''[T]$ from $H'[T]$, for every pair $(e, e') \in B$ with $e, e' \in H'[T]$, we delete the first edge $e$, and let $H''[T]$ be the remaining subgraph.

\item \textbf{(Construction of $H'''[T]$)}
To construct $H'''[T]$ from $H''[T]$, we delete all edges in $T \setminus \mst(H''[T])$.
\end{itemize}

The important properties of these subgraphs are:
\begin{lemma} \label{lem:highwg}
Every graph $H'''[T]$ has weighted girth $>k+1$ (deterministically).
\end{lemma}
\begin{proof}
Let $C$ be any cycle in $H[T]$ of normalized weight $\le k+1$.
We will argue that $C$ does not survive in $H'''[T]$.
There are two cases, depending on whether $C$ has a heaviest edge $(u, v) \notin T$, or whether all heaviest edges in $C$ are in $T$.

In the first case, suppose that there is a heaviest edge $(u, v) \in C, (u, v) \notin Q$.
By the properties of the blocking set $B$ from Lemma \ref{lem:bsetexists}, there is $(e, e') \in B$ with $e, e' \in C$ and $e \notin T$.
It is not possible for both edges $e, e'$ to survive in $H''[T]$, since if they are both selected to remain in $H'[T]$ then we will choose to delete $e$ when we construct $H''[T]$.
Thus $C$ does not survive in $H''[T]$.

In the second case, suppose that all heaviest edges in $C$ are also in $Q$.
Suppose that all edges in $C$ survive in $H''[T]$.
Then at least one of the heaviest edges $(u, v) \in C$ will not be in $\mst(H''[T])$ (since otherwise we could exchange it for a lighter edge on $C$).
So we will remove $(u, v)$ when we move from $H''[T]$ to $H'''[T]$.
\end{proof}

\begin{lemma} \label{lem:exphweight}
Every graph $H'''[T]$ has expected weight
$$\mathbb{E}\left[w(H'''[T])\right] \ge \Omega\left( \frac{w(H[T])}{f} \right) - w(T).$$
\end{lemma}
\begin{proof}
Every edge in $T$ is included in $H''[T]$ deterministically.
Every edge $e \in E(H[T]) \setminus T$ survives in $H''[T]$ iff the following two independent events both occur:
\begin{itemize}
\item $e$ is selected to be included in $H'[T]$, which happens with probability $1/f$, and
\item For all pairs $(e, e') \in B$ with $e$ first, the other edge $e'$ is \emph{not} selected to be included in $H'[T]$.
Since $B$ is $f$-capped, there are at most $f$ such edges $e'$ to consider.
Additionally, since $T$ hosts $e$, none of these edges $e'$ are in $T$ itself.
There are at most $f$ such blockers, so the probability that none is selected is at least $(1-1/f)^f\ge1/4$.

\end{itemize}
Thus each edge $e \in E(H[T])$ survives in $H''[T]$ with probability $\Omega(1/f)$, implying that
$$\mathbb{E}\left[w(H''[T])\right] \ge \Omega\left( \frac{w(H[T])}{f} \right).$$
Finally, when we move from $H''[T]$ to $H'''[T]$, we can only possibly delete edges in $T$ and so we delete at most $w(T)$ weight, implying the lemma.
\end{proof}

\paragraph{Analysis of Lightness.}

As in the probabilistic method, there exists a realization of the subgraphs $H'''[T]$ that satisfies the expected weight inequality in the previous lemma (for all $T$).
Rearranging the bound from Lemma \ref{lem:exphweight}, we have:
\begin{align*}
w(H'''[T]) &\ge \Omega\left( \frac{w(H[T])}{f} \right) - w(T)\\
\implies O(f) \cdot \left(w(H'''[T]) + w(T) \right) &\ge w(H[T]).
\end{align*}
Using this, we observe that for all trees $T \in \tee$, we have 
\begin{align}
\frac{w(H[T])}{w(T)} &\le \frac{O\left(f\right) \cdot \left(w(H'''[T]) + w(T)\right)}{w(T)} \notag \\
&\le O(f) \cdot \frac{w(H'''[T])}{w(\mst(H'''[T]))} + O(f)\notag \\
&= O(f) \cdot \ell(H'''[T]) + O(f)\notag \\ 
&\le O(f) \cdot \lambda(n, k+1) \label{eq:single-bound}
\end{align}
where the last steps follow from the definition of $\lambda$, the fact that $\lambda(n, k+1) \ge 1$, and the previous lemma establishing that $H'''[T]$ has weighted girth $>k+1$.
If $H=Q$, the lightness is $1$ and the bound is immediate. Otherwise, a hosted edge guarantees $\sum_{T\in\tee}w(T)>0$. Thus, to wrap up the proof, we bound:
\begin{align*}
\ell_{2f}(H \mid G) &= 1+\frac{w(H\setminus Q)}{w(Q)}\\
&\le 1+\frac{\sum \limits_{T\in\tee}w(H[T])}{w(Q)}
\tag{every $e\in E(H)\setminus Q$ has a host} \\ 
&\leq 1+2 \cdot \frac{\sum \limits_{T \in \tee} w(H[T])}{\sum \limits_{T \in \tee} w(T)} \tag{Corollary~\ref{cor:pairwisenw} and def of $\tee$}\\
&\le 1+O\left( \frac{\sum \limits_{T \in \tee} f \cdot \lambda(n, k+1) \cdot w(T)}{\sum \limits_{T \in \tee} w(T)} \right) \tag{Eq.~\eqref{eq:single-bound}}\\
&= O(f \cdot \lambda(n, k+1)).
\end{align*}

\subsection{Improved Lightness Bounds via Multiple Host Trees} \label{sec:multiple-host}

A way in which our previous analyses are suboptimal is that they do not acknowledge the possibility that an edge could be hosted by many trees, not just one.
As a simple example to introduce the technique, let us observe that the lightness bound from Algorithm \ref{alg:ftgreedy} improves by a factor of $f$ if we are willing to take a higher lightness competition parameter:
\begin{theorem} \label{thm:high-parameter-lightness}
Suppose we modify Algorithm \ref{alg:ftgreedy} to instead construct $Q$ as a min-weight $(2+\eta)f$-FT connectivity preservers, for some $\eta >0$.
Then the output $f$-EFT spanner $H$ has $(2+\eta)f$-competitive lightness
$$ \ell_{(2+\eta)f}(H \mid G) \le 1+O\left( \eta^{-1} \lambda(n, k+1)\right).$$
\end{theorem}
\begin{proof}
Write $r=\lfloor(2+\eta)f\rfloor$. Follow the previous analysis, but note that our guarantee is now that we have a set of trees $\tee$ such that (1) every edge is in at most two trees, and (2) for all $(u, v) \in E(G) \setminus Q$, there are at least $r+1$ trees in $\tee$ that span the $(r+1)$-connected component that contains $u, v$.
Thus, letting $F$ be a fault set that caused us to add $(u, v)$ to the spanner, there are at least $r+1-2f=\lfloor\eta f\rfloor+1>\eta f$ of these trees that are disjoint from $F$.
We let \textbf{each} of these trees host the edge $(u, v)$.
In particular, this means that $(u, v)$ will appear in at least $\eta f$ many $H[T]$ subgraphs.
From there the analysis continues as before, but our final calculation is as follows.
The case $H=Q$ is immediate, so assume $H\ne Q$.
\begin{align*}
\ell_{(2+\eta)f}(H \mid G) &= 1+\frac{w(H\setminus Q)}{w(Q)}\\
&\leq 1+O\left(\frac{(\eta f)^{-1} \sum \limits_{T \in \tee} w(H[T])}{w(Q)}\right) \tag{every $e \in E(H)\setminus Q$ has at least $\eta f$ hosts} \\ 
&\leq 1+O\left(\frac{(\eta f)^{-1} \sum \limits_{T \in \tee} w(H[T])}{\sum \limits_{T \in \tee} w(T)}\right) \tag{Corollary~\ref{cor:pairwisenw} and def of $\tee$}\\
&\le 1+O\left( \frac{(\eta f)^{-1} \sum \limits_{T \in \tee} f \cdot \lambda(n, k+1) \cdot w(T)}{\sum \limits_{T \in \tee} w(T)} \right) \tag{Eq.~\eqref{eq:single-bound}}\\
&= 1+O(\eta^{-1} \cdot \lambda(n, k+1)). \tag*{\qedhere}
\end{align*}
\end{proof}

If we don't wish to harm our competition parameter, we can use a related method to improve the dependence to $f^{1/2}$:
\begin{theorem} \label{thm:warmupgoodf}
The output spanner $H$ from Algorithm \ref{alg:ftgreedy} (with no modifications) satisfies
$$\ell_{2f}(H \mid G) \le O\left( f^{1/2} \cdot \lambda(n, k+1) \right).$$
\end{theorem}
\begin{proof}
Let us say that an edge $e \in E(H) \setminus Q$ is \emph{$Q$-heavy} if there are at least $f - f^{1/2}$ different edges $e' \in Q$ with $(e, e') \in B$.
Otherwise, we say that $e$ is \emph{$Q$-light}.
Let $H_{heavy} \subseteq H$ be the subgraph that includes $Q$ and all edges in $E(H) \setminus Q$ are $Q$-heavy, and similarly let $H_{light} \subseteq H$ be the subgraph that includes $Q$ and all edges in $E(H) \setminus Q$ that are $Q$-light.
Notice that
\begin{align*}
\ell_{2f}(H) &= \frac{w(H \setminus Q) + w(Q)}{w(Q)} \\
&= \frac{w(H_{heavy} \setminus Q) + w(H_{light} \setminus Q) + w(Q)}{w(Q)} \\
&= \frac{w(H_{heavy})}{w(Q)} + \frac{w(H_{light})}{w(Q)} - 1\\
&< \ell_{2f}(H_{heavy}) + \ell_{2f}(H_{light}).
\end{align*}
Thus it suffices to individually bound $\ell_{2f}(H_{heavy}), \ell_{2f}(H_{light})$.

\paragraph{Bounding $\ell_{2f}(H_{heavy})$.}

Follow the same argument as before, but observe that any $Q$-heavy edge $e$ can be blocked with at most $f^{1/2}$ other edges $e'$ in its relevant subgraph $H[T]$.
Thus we may generate $H'[T]$ from $H[T]$ by keeping all edges in $T$ deterministically (as before), but by sampling the edges in $H[T] \setminus T$ with probability only $1/f^{1/2}$ instead of $1/f$.
This leads to a new bound of
$$\mathbb{E}\left[w(H''[T])\right] \ge \Omega\left( \frac{w(H[T])}{f^{1/2}} \right),$$
and following the same analysis from there, this improved factor of $f^{1/2}$ propagates into the final lightness bound.

\paragraph{Bounding $\ell_{2f}(H_{light})$.}

Follow the same argument as before, but observe that for any $Q$-light edge $e$, there are at least 
$$2f+1 - 2(f - f^{1/2}) \ge f^{1/2} +1$$ 
different trees that span the $(2f+1)$-connected component containing the endpoints of $e$, and which do not contain edges $e'$ with $(e, e') \in B$.
We let $e$ be hosted by \textbf{all} such trees.
We then continue the analysis as before, with a slight change only in the final calculation, much like the previous theorem.
If $H_{light}=Q$, its lightness is $1$; otherwise the denominator below is positive and the calculation is now:
\begin{align*}
\ell_{2f}(H_{light}) &=1+\frac{w(H_{light}\setminus Q)}{w(Q)}\\
&\le 1+O\left(\frac{f^{-1/2}\sum \limits_{T\in\tee}w(H[T])}{\sum \limits_{T\in\tee}w(T)}\right)
\tag{each non-$Q$ edge has at least $f^{1/2}$ hosts}\\
&\le 1+O\left( \frac{f^{-1/2} \sum \limits_{T \in \tee} f \cdot \lambda(n, k+1) \cdot w(T)}{\sum \limits_{T \in \tee} w(T)} \right) = O\left(f^{1/2} \cdot \lambda(n, k+1)\right). \tag*{\qedhere}
\end{align*}
\end{proof}

These two theorems now essentially directly imply the upper bound parts of Theorems~\ref{thm:introsmallcomp} and \ref{thm:introbigcomp}:

\begin{proof}[Proof of Theorems~\ref{thm:introsmallcomp} and \ref{thm:introbigcomp}, Upper Bounds]
Apply the following results with algorithm stretch
$$s=(1+\eps)2k-1=(1+O(\eps))(2k-1),$$
so that $s+1=(1+\eps)2k$.
If $\eps>1$, obtain the coarse $O_{\eps}(n^{1/k})$ estimate by applying Theorem~\ref{thm:priorlight} at $\eps=1$.

    Theorem~\ref{thm:introsmallcomp} is directly implied by Theorem~\ref{thm:feasible} (which shows that $H$ is indeed an $f$-EFT $k$-spanner) and Theorem~\ref{thm:warmupgoodf} (which gives the claimed lightness bound).  For Theorem~\ref{thm:introbigcomp}, it is easy to see that the change we made in the algorithm (making $Q$ a $(2+\eta)f$-EFT connectivity preserver rather than $2f$) does not affect the proof of Theorem~\ref{thm:feasible} that $H$ is feasible.  Hence Theorem~\ref{thm:introbigcomp} is implied by Theorem~\ref{thm:feasible} and Theorem~\ref{thm:high-parameter-lightness}.
\end{proof}

\ifshort
\else
\subsection{Competitive Lightness Upper Bounds in Polynomial Time} \label{sec:polytime}

\begin{algorithm}[t]
\DontPrintSemicolon

\textbf{Input:} Graph $G$, stretch $k$, fault tolerance $f$\;~\\

Let $Q \gets $ $2$-approximate min-weight $2f$-FT connectivity preserver of $G$ \tcp{polytime by \cite{DKK22}}

Let $\tee \gets$ tree decomposition of $Q$ satisfying Corollary \ref{cor:pairwisenw} w.r.t.\ connectivity parameter $2f+1$\;

Let $E \gets \emptyset$ be the initial set of non-$Q$ spanner edges\;~\\

\tcp{$c$ large enough constant}
\ForEach{edge $(u, v) \in E(G) \setminus Q$ in order of nondecreasing weight}{

    \ForEach{tree $T \in \tee$ that spans the $(2f+1)$-connected component of $Q$ containing $u, v$}{
        sample $c \log n$ subgraphs by including $T$, and including each edge in $E$ with probability $1/f$\;

        let $\widehat{P}^T_{(u, v)}$ be the fraction of sampled subgraphs $H'$ in which $\dist_{H'}(u, v) > k \cdot w(u, v)$\;
    
        \If{$\widehat{P}^T_{(u, v)} \ge 1/8$}{
            add $(u, v)$ to $E$\;
        }
    }
}
\textbf{Return} $H = Q \cup E$\;
\caption{\label{alg:ftpoly} FT spanners in polynomial time}
\end{algorithm}

We now prove Theorem~\ref{thm:alg-polytime}, improving to polynomial runtime at the cost of a worse dependence on $f$ (relative to the previous section).
We may again assume without loss of generality that $f$ is at least a sufficiently large absolute constant.
Recall the theorem:

\algpolytime*

The algorithm that we will analyze to prove this is Algorithm~\ref{alg:ftpoly}.  Both the algorithm and its analysis are adaptations of the ideas of \cite{BDR21}.
To begin the analysis, let us set up some useful definitions.
For an edge $e \in E(G) \setminus Q$, let us write $E_e$ for the subset of non-$Q$ spanner edges that were added strictly before the edge $e$ was considered in the algorithm.
We also write
$$H_e := Q \cup E_e \qquad \text{and} \qquad H_e^T := T \cup E_e$$
for trees $T \in \tee$.
For all $e \in E(G) \setminus Q, T \in \tee$, we let $H'^T_e$ be a random subgraph of $H_e^T$ obtained by including the edges of $T$ (deterministically), and then including each other edge from $H_e^T$ independently with probability $1/f$.
Then let
$$P^T_{e=(u, v)} := \Pr\left[ \dist_{H'^T_e}(u, v) > k \cdot w(u, v) \right]$$
where the probability is over the random definition of $H'^T_e$.
We cannot compute $P^T_{e=(u, v)}$ exactly, but we can view Algorithm \ref{alg:ftpoly} as computing experimental estimates $\widehat{P}^T_{e=(u, v)}$ by repeatedly sampling subgraphs $H'^T_e$.
The following lemma applies Hoeffding's inequality~\cite{Hoeffding63} to show that our estimates are probably reasonably good.
This argument directly follows one from \cite{BDR21}, but we recap it here from scratch.

\begin{lemma} [c.f.~\cite{BDR21}, Lemma 3.1] \label{lem:whpguarantee}
With high probability, for every estimate $\widehat P_e^T$ evaluated before $e$ is added to $E$, we have
$\widehat{P}^T_e \in P^T_e \pm \frac{1}{16}.$
\end{lemma}
\begin{proof}
Condition on the algorithm history immediately before the current edge $e$ is considered. For every relevant test (before $e$ is inserted), the sampled edge set is the fixed set $E_e$. Let $N=\lceil c\log n\rceil$ and let $X_1,\ldots,X_N\in\{0,1\}$ indicate whether the independently sampled subgraphs have distance greater than $k\,w(e)$. Then
$$\widehat P_e^T=\frac1N\sum_{i=1}^N X_i,
\qquad \mathbb E[X_i]=P_e^T.$$
Hoeffding's inequality~\cite{Hoeffding63} gives
\begin{align*}
\Pr\left[\left|\widehat P_e^T-P_e^T\right|>\frac1{16}\right]
&\le 2\exp(-N/128)\le 2n^{-c/128}.
\end{align*}

A union bound over the polynomially many edge--tree tests gives the claimed high-probability event for sufficiently large $c$.
\end{proof}

The following lemmas will assume that the event from Lemma \ref{lem:whpguarantee} holds, and hence these lemmas only hold with high probability, rather than deterministically.
We next establish correctness:
\begin{lemma} \label{lem:poly-correct}
With high probability, the output spanner $H$ from Algorithm \ref{alg:ftpoly} is an $f$-EFT $k$-spanner of the input graph.
\end{lemma}
\begin{proof}
Consider an edge $(u, v) \in E(G) \setminus Q$, and note that as in the previous warmup, this implies that $u, v$ are $2f+1$-connected.
It suffices to argue that, if there is a fault set $F, |F| \le f$ for which 
$$\dist_{H_e \setminus F}(u, v) > k \cdot w(u, v),$$
then we choose to add $(u, v)$ to the spanner.
Suppose such a fault set $F$ exists.
Recall that there are $2f+1$ trees that span the $2f+1$-connected component that contains $u, v$.
Since each edge in $F$ belongs to at most two of these trees, there exists a tree $T \in \mathcal{T}$ such that $u, v$ are connected in $T$ and $T \cap F = \emptyset$.
When we generate a subgraph $H'$ associated to this particular tree $T$, note that if none of the edges in $F$ survive in $H'$, then we will have
$$\dist_{H'}(u, v) > k \cdot w(u, v).$$
There are $|F|\le f$ edges, which are each included in $H'$ with probability $1/f$, and so none of the edges in $F$ survive in $H'$ with probability at least $1/4$.
So we have $P_e^T \ge 1/4$.
Assuming that the event from Lemma \ref{lem:whpguarantee} holds (relating $P^T_e$ to $\widehat{P}_e^T$), we thus have $\widehat{P}_e^T \ge 1/4-1/16=3/16>1/8$.
So $(u, v)$ will be added to the spanner $H$ when the tree $T$ is considered, completing the proof.
\end{proof}

We now turn to bounding the number of edges in the final spanner.
For an edge $(u, v) \in E$, we will say it is \emph{hosted} by the first tree $T \in \tee$ that caused the edge to be added to $E$ in the main loop.
We will write $H[T]$ for the subgraph of the final spanner that contains $T$ and all of the edges that it hosts.
Note that we are reusing this notation from the previous warmup; although the definition of $H[T]$ is slightly different here, it plays a directly analogous role.

\begin{lemma} \label{lem:polyhostbound}
With high probability, for all trees $T \in \mathcal{T}$, we have
$$\frac{w(H[T])}{w(T)} \le O\left(f \cdot \lambda(n, k+1)\right).$$
\end{lemma}
\begin{proof}
Analogous to the previous warmup, it will be helpful to define subgraphs $$H'''[T] \subseteq H''[T] \subseteq H'[T] \subseteq H[T]$$
as follows.
First, $H'[T]$ is a random subgraph of $H[T]$, obtained by including $T$ deterministically and then including each edge hosted by $T$ with probability $1/f$.
For an edge $(u, v) \in H'[T]$, let $H'_{(u, v)}[T]$ be the subgraph that contains only $T$ and the edges from $H'[T]$ considered strictly before $(u, v)$ in the algorithm.
Then we define $H''[T]$ as:
$$H''[T] := T \cup \left\{ (u, v) \in H'[T] \ \mid \ \dist_{H'_{(u, v)}[T]}(u, v) > k \cdot w(u, v)\right\}.$$
Finally, we define
$$H'''[T] := H''[T] \setminus \left( T \setminus \mst(H''[T])\right).$$
We can now bound the expected weight of $H'''[T]$ in two different ways:
\begin{itemize}
\item On one hand, the definition of $H'''[T]$ implies that it has weighted girth $>k+1$.
The argument is analogous to Lemma \ref{lem:highwg}.
Consider a cycle $C$ in $H'[T]$ of normalized weight $\le k+1$.
We consider two cases, by whether or not there is a heaviest edge $(u, v) \in C$ that is not also in $T$.

First suppose that there is a heaviest edge $(u, v) \in C, (u, v) \notin T$, and let $(u, v)$ specifically be the last such edge considered by the algorithm.
If all other edges in $C \setminus \{(u, v)\}$ are added to $H''[T]$ then then they form a $u \leadsto v$ path, implying that
$$\dist_{H'_{(u, v)}[T]}(u, v) \le k \cdot w(u, v),$$
which will cause us to not include $(u, v)$ in $H''[T]$.

On the other hand, suppose that all heaviest edges $(u, v) \in C$ are also in $T$.
Then $\mst(H''[T])$ will omit at least one such edge (since otherwise we could exchange it for a lighter edge on $C$), so $C$ will not survive in $H'''[T]$.

Putting these together: since $H'''[T]$ has weighted girth $>k+1$, by definition of $\lambda$, we have
\begin{align*}
\lambda(n, k+1) \ge \ell(H'''[T]) &= \frac{w(H'''[T])}{w(\mst(H'''[T]))}\\
&\ge \frac{w(H'''[T])}{w(T)}.
\end{align*}

\item On the other hand, consider an edge $(u, v)$ hosted by $T$, and let us check the probability that it survives in $H''[T]$.
First it survives in $H'[T]$ with probability $1/f$.
Then, since $H_{(u, v)}[T] \subseteq H_{(u, v)}^T$,\footnote{As a reminder, the difference between these two subgraphs is that $H_{(u, v)}[T]$ only contains the spanner edges considered before $(u, v)$ that are hosted by $T$, while $H_{(u, v)}^T$ contains all spanner edges considered before $(u, v)$.} it survives in $H''[T]$ with probability \emph{at least} $P_{(u, v)}^T$.
Since $T$ caused $(u,v)$ to be added to the spanner, we have $\widehat{P}_{(u, v)}^T \ge 1/8$, and assuming that the event from Lemma \ref{lem:whpguarantee} holds, we thus have
$$P_{(u, v)}^T \ge \frac{1}{8}-\frac{1}{16}=\frac{1}{16}.$$
So in total, $(u, v)$ survives in $H''[T]$ with probability at least $1/(16f)=\Omega(1/f)$, and we have
$$\mathbb{E}\left[w(H''[T])\right] \ge \Omega\left(\frac{w(H[T])}{f} \right).$$

We then delete at most $w(T)$ additional edge weight when we move from $H''[T]$ to $H'''[T]$, so we have
$$\mathbb{E}\left[w(H'''[T])\right] \ge \Omega\left(\frac{w(H[T])}{f}\right)-w(T).$$
\end{itemize}

Combining the previous inequalities, we have:
\begin{align*}
w(H[T]) &\le O(f)\left(\mathbb{E}\left[w(H'''[T])\right]+w(T)\right)\\
&\le O\left(f \cdot \lambda(n, k+1) \cdot w(T)\right). \tag*{\qedhere}
\end{align*}
\end{proof}

We can now wrap up the proof:

\begin{lemma} \label{lem:poly-light}
With high probability, the output spanner $H$ has $2f$-competitive lightness
$$\ell_{2f}(H \mid G) \le O\left( f \cdot \lambda(n, k+1) \right).$$
\end{lemma}
\begin{proof}
If $H=Q$, its competitive lightness is at most $2$; otherwise the denominator below is positive. We have:
\begin{align*}
\ell_{2f}(H \mid G) &\le 2+2\cdot\frac{w(H\setminus Q)}{w(Q)} \tag{$Q$ is $2$-approx min weight preserver}\\
&\leq 2+2 \cdot \frac{\sum \limits_{T \in \tee} w(H[T])}{w(Q)} \tag{every $e \in E(H)\setminus Q$ in at least one $H[T]$} \\ 
&\leq 2+4 \cdot \frac{\sum \limits_{T \in \tee} w(H[T])}{\sum \limits_{T \in \tee} w(T)} \tag{Corollary~\ref{cor:pairwisenw} and def of $\tee$}\\
&\le 2+O\left( \frac{\sum \limits_{T \in \tee} f \cdot \lambda(n, k+1) \cdot w(T)}{\sum \limits_{T \in \tee} w(T)} \right) \tag{Lemma \ref{lem:polyhostbound}}\\
&= O(f \cdot \lambda(n, k+1)). \tag*{\qedhere}
\end{align*}
\end{proof}

The first part of Theorem~\ref{thm:alg-polytime} is now directly implied by Lemmas~\ref{lem:poly-correct} and \ref{lem:poly-light}.

\subsubsection{Improved Lightness Bounds in Polynomial Time via Multiple Host Trees}

Of the two results from Section \ref{sec:multiple-host}, only one of them extends readily to the polynomial time algorithm.
The one that does \emph{not} seem to extend is Theorem \ref{thm:warmupgoodf}, which improves the $f$-dependence for $2f$-competitive lightness to $f^{1/2}$.
The issue is roughly that the proof samples heavy and light edges with different probabilities, but in our polynomial time algorithm we commit to a sampling probability at runtime, meaning that we cannot easily apply different runtimes to different edge types.
However, Theorem \ref{thm:high-parameter-lightness} extends fairly straightforwardly, in order to prove the corresponding part of Theorem \ref{thm:alg-polytime}.

\begin{algorithm}[t]
\DontPrintSemicolon

\textbf{Input:} Graph $G$, stretch $k$, fault tolerance $f$, parameter $\eta > 0$\;~\\

Let $r\gets\lfloor(2+\eta)f\rfloor$\;

Let $Q \gets $ $2$-approximate min-weight $r$-FT connectivity preserver of $G$ \tcp{polytime by \cite{DKK22}}

Let $\tee \gets$ tree decomposition of $Q$ satisfying Corollary \ref{cor:pairwisenw} w.r.t.\ connectivity parameter $r+1$\;

Let $E \gets \emptyset$ be the initial set of non-$Q$ spanner edges\;~\\

\tcp{$c$ large enough constant}
\ForEach{edge $(u, v) \in E(G) \setminus Q$ in order of nondecreasing weight}{
    votes $\gets 0$\;
    \ForEach{tree $T \in \tee$ that spans the $(r+1)$-connected component of $Q$ containing $u, v$}{
        sample $c \log n$ subgraphs by including $T$, and including each edge in $E$ with probability $1/f$\;

        let $\widehat{P}^T_{(u, v)}$ be the fraction of sampled subgraphs $H'$ in which $\dist_{H'}(u, v) > k \cdot w(u, v)$\;
    
        \If{$\widehat{P}^T_{(u, v)} \ge 1/8$}{
            votes $\gets$ votes $+1$\;
        }
    }
    \If{votes $\ge r+1-2f$}{
        add $(u, v)$ to $E$\;
    }
}
\textbf{Return} $H = Q \cup E$\;
\caption{\label{alg:polymultitree} FT spanners in polynomial time with competition parameter $(2+\eta)f$}
\end{algorithm}

To prove this we will use Algorithm~\ref{alg:polymultitree}.  Most of our analysis of Algorithm \ref{alg:polymultitree} is the same as that for Algorithm \ref{alg:ftpoly}, and we will not repeat it here.
For example, Lemma \ref{lem:whpguarantee} (asserting that $\widehat{P}^T_{(u , v)} \approx P^T_{(u, v)}$ for all $T, (u, v)$) still holds with exactly the same proof.
Lemma \ref{lem:poly-correct} (asserting that the output spanner is correct with high probability) also still holds by essentially the same proof: for any edge $e$ and a fault set $F$, $|F|\le f$, witnessing the distance condition in that proof, there are at least $r+1-2f=\lfloor\eta f\rfloor+1$ eligible trees disjoint from $F$. For each such tree $T$, we have $P_e^T\ge1/4$ and hence $\widehat P_e^T\ge3/16>1/8$ by Lemma~\ref{lem:whpguarantee}.
All such trees will vote for $(u, v)$ (i.e. we will increment the ``votes'' variable when these trees are considered), and so $(u, v)$ will indeed be added to the spanner if an appropriate fault set $F$ exists.

The part that changes is the hosting of edges in trees: instead of an edge $(u, v) \in E$ being hosted by the \emph{first} tree $T$ that caused the edge to be added to the spanner, it is instead hosted by \emph{all} trees that voted for it, of which there are at least $r+1-2f=\lfloor\eta f\rfloor+1>\eta f$.
Under this new hosting strategy, the proof of Lemma \ref{lem:polyhostbound} (controlling the weight of each host graph $H[T]$) still holds with no significant changes.
But the final calculation in Lemma \ref{lem:poly-light} now admits an optimization.
We calculate:

If $H=Q$, its competitive lightness is at most $2$; otherwise the denominator below is positive.
\begin{align*}
\ell_{(2+\eta)f}(H \mid G) &\le 2+2\cdot\frac{w(H\setminus Q)}{w(Q)} \tag{$Q$ is $2$-approx min weight preserver}\\
&\leq 2+2 \cdot \frac{(\eta f)^{-1} \sum \limits_{T \in \tee} w(H[T])}{w(Q)} \tag{every $e \in E(H)\setminus Q$ in at least $\Omega(\eta f)$ graphs $H[T]$} \\ 
&\leq 2+4 \cdot \frac{(\eta f)^{-1} \sum \limits_{T \in \tee} w(H[T])}{\sum \limits_{T \in \tee} w(T)} \tag{Corollary~\ref{cor:pairwisenw} and def of $\tee$}\\
&\le 2+O\left( \frac{(\eta f)^{-1}\sum \limits_{T \in \tee} f \cdot \lambda(n, k+1) \cdot w(T)}{\sum \limits_{T \in \tee} w(T)} \right) \tag{Lemma \ref{lem:polyhostbound}}\\
&= 2+O\left(\eta^{-1} \cdot \lambda(n, k+1)\right). \tag*{\qedhere}
\end{align*}

\fi

\ifshort \else
\section{Lower Bounds} \label{sec:lowerbound}

In this section we prove our lower bounds.  We begin with our main lower bound, showing that if we want lightness bounds similar to the non-fault tolerant setting (e.g., dependence of $n^{1/k}$ for a $(1+\eps)\cdot (2k-1)$-spanner, or even anything below $n/k$) then we need the competition parameter to be at least $2f$.  We then give lower bounds for $2f$-competitive lightness.

\subsection{Competition Parameters at most $2f-1$} \label{sec:lowersmallcompetitive}
We now want to prove Theorem~\ref{thm:lower-bound-main}, which essentially rules out sublinear $2f-1$-competitive lightness (for small $f, k$): 

\bicriterialower*

Before proving our main lower bound, we start with a warmup in the setting of a single fault to show a lightness bound of $\Omega(n/k)$. Let $G$ be a graph of $2n$ vertices $v_0$, $v_1$, \ldots, $v_{2n-1}$ with unit weight edges from $v_i$ to $v_{i+1} \pmod{2n}$  and additional edges of weight $\frac{2n-2}{k} - \epsilon$ from $v_{2j}$ to $v_{2j+2}$. Clearly $\mst(G)$ is the cycle of unit weight edges (minus one arbitrary edge), and so has weight $2n-1$.  On the other hand, we claim that any $1$-EFT $k$-spanner $H$ must include all of the additional edges.  To see this, suppose that $H$ does not include one of the additional edges, say $\{v_{2j}, v_{2j+2}\}$.  Consider the fault set $F = \{\{v_{2j}, v_{2j+1}\}\}$.  Then $\dist_{G \setminus F}(v_{2j}, v_{2j+2}) = \frac{2n-2}{k} - \epsilon$, via the edge $\{v_{2j}, v_{2j+2}\}$.  But in $H \setminus F$, the only way to get from $v_{2j}$ to $v_{2j+2}$ is to go all the way around the cycle, for a total distance of $2n-2$.  So then the stretch is $(2n-2) / \left(\frac{2n-2}{k} - \epsilon\right) > k$, and thus $H$ is not a $1$-EFT $k$-spanner.  

Thus $H$ has all of the additional edges, and so has total weight at least $n \cdot \left(\frac{2n-2}{k} - \epsilon\right) = \Omega(n^2 / k)$.  Thus $\ell_1(H \mid G) \geq \Omega(n/k)$.

\begin{figure}[h]
    \centering
    \includegraphics[scale = 0.08]{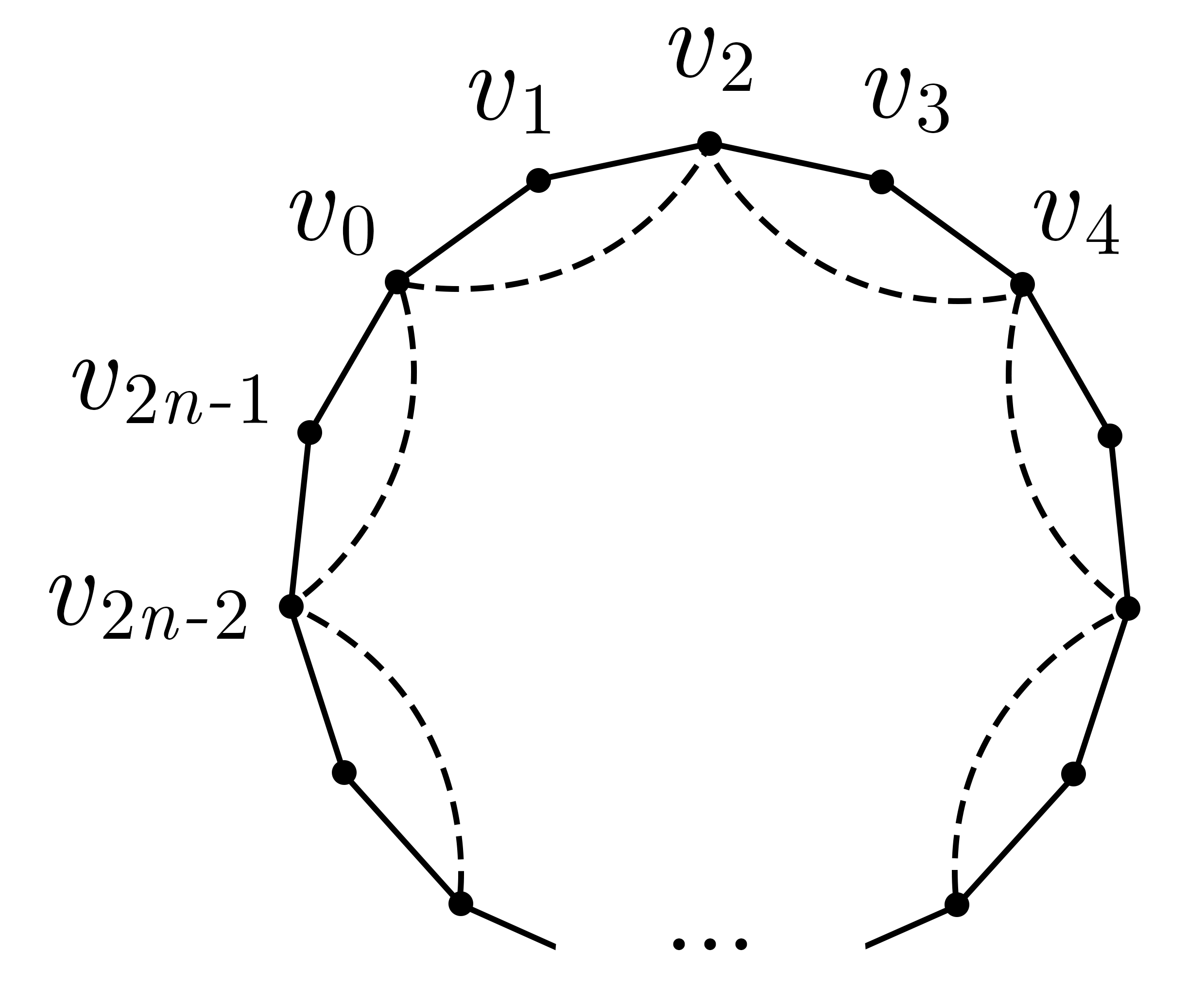}
    \caption{Outer (solid) cycle edges have unit weight, dotted edges have weight $\frac{2n-2}{k} - \epsilon$}
    \label{fig:enter-label}
\end{figure}

To extend this to larger faults (and thus prove Theorem~\ref{thm:lower-bound-main}). we modify the above argument by replacing each of the $v_{2j+1}$ nodes (the odd indices) by a ``cloud'' of $f$ nodes.

\begin{proof}[Proof of Theorem~\ref{thm:lower-bound-main}]
    Let $G$ be a graph of $m+mf$ vertices $v_0$, $v_1$, \ldots, $v_{m-1}$ and $v_{i,j}$ for $i \in \{0, 1, \ldots m-1\}$, $j \in [f]$, with edges of unit weight from $v_i$ to all $v_{i,j}$ and from $v_{i+1} \pmod{m}$ to all $v_{i,j}$, and additional edges of weight $\frac{2m-2}{k} - \epsilon$ from $v_i$ to $v_{i+1} \pmod{m}$. Then it is not hard to see that the set of all unit weight edges forms a $(2f-1)$-EFT connectivity preserver, and has weight $2mf$. However, any $f$-EFT $k$-spanner $H$ must include every heavier edge; suppose $H$ does not include some additional edge $\{v_i, v_{i+1}\}$. Consider the fault set $F = \{\{v_i, v_{i,j}\}: j \in [f]\}$. Then we have $\dist_{G\setminus F} (v_i, v_{i+1}) = \frac{2m-2}{k} - \epsilon$ using the edge $\{v_i, v_{i+1}\}$, but in $H \setminus F$ the minimum possible distance between $v_i$ and $v_{i+1}$ would be $2m-2$ by traversing around the cycle using unit weight edges, which again exceeds the stretch factor $k$. Therefore, $H$ contains all the heavier additional edges, and has weight at least $m\cdot (\frac{2m-2}{k} - \epsilon) = \Omega(m^2/k)$. It follows that $\ell_{2f-1}(H \mid G) \geq \Omega(m/fk) = \Omega(n/f^2k)$. 
\end{proof}

\subsection{$2f$- and $2(1+\eta)f$-Competitive Lightness} \label{sec:lowerbigcompetitive}
We now prove the lower bound parts of Theorem~\ref{thm:introsmallcomp} and Theorem~\ref{thm:introbigcomp}.  We in fact prove a stronger statement which directly implies both of them.

\begin{restatable}{theorem}{lowergraph} \label{thm:lower-bound-lambda-graph}
     For any constant $c \geq 2$, there is an infinite family of $n$-node graphs $G$ for which every $f$-EFT $k$-spanner $H$ has competitive lightness
     $$\ell_{cf}(H \mid G) \geq \Omega\left(\lambda\left(\frac{n}{cf}, k+1\right)\right).$$
\end{restatable}
\begin{proof}
    Let $G' = (V', E')$ be a connected graph on $n' \geq 2$ nodes with weighted girth greater than $k+1$ and lightness $\Omega(\lambda(n', k+1))$, and let $T'$ be its MST.

    We create $G = (V, E)$ as follows.  Let $p = \lfloor cf \rfloor \geq 2f$.  Let $V = V' \times [p]$, so $n=|V| = n'p$.  For each $v \in V'$, we refer to $C_v = \{(v, i) : i \in [p]\}$ as the ``cloud'' of $v$.  For each $e' = \{u,v\} \in E'$, add all $p^2$ edges between $C_u$ and $C_v$, each with weight $w(e')$.  Within each cloud, also add a clique whose edges have weight $\delta = w(T')/(n'p^2)$.  All internal cliques together have weight less than $w(T')/2$.

    \begin{figure}[ht]
    \centering
    \includegraphics[width=0.92\linewidth]{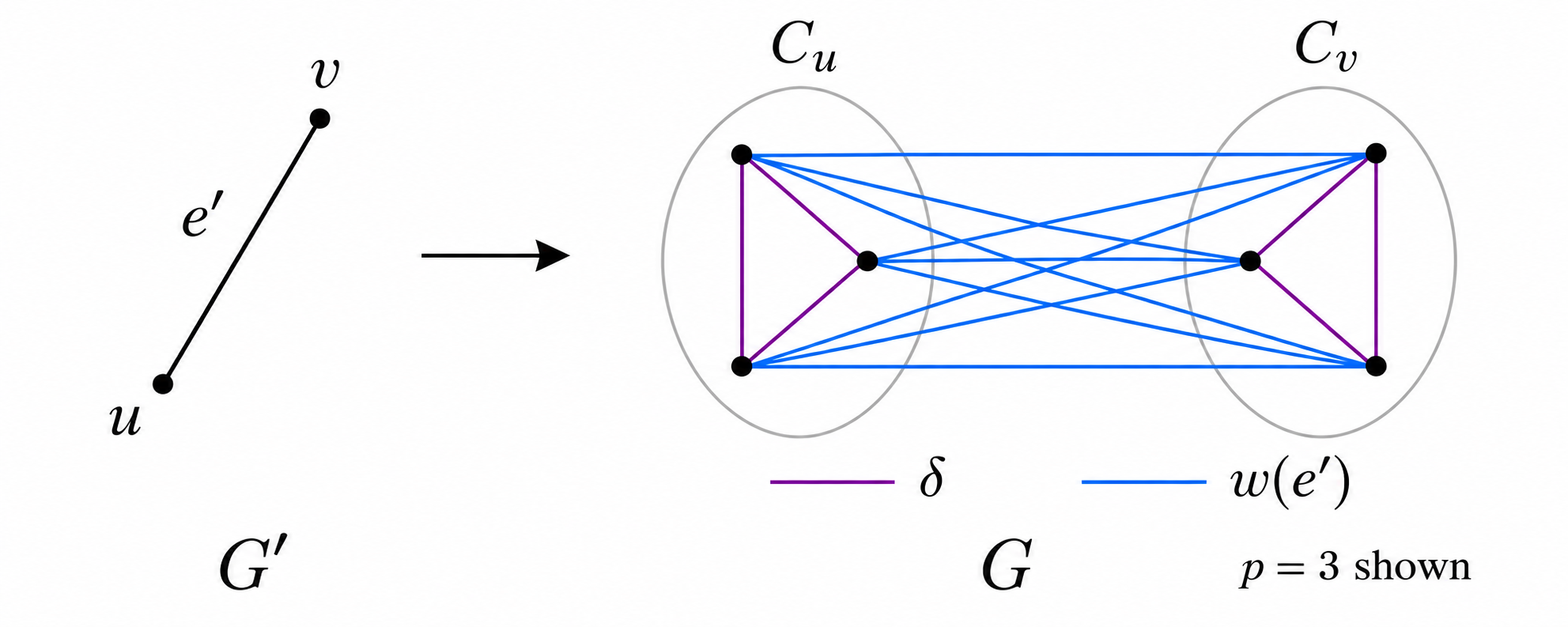}
    \caption{Cloud blowup of a core edge $e'=\{u,v\}$, illustrated with $p=3$. Each endpoint is replaced by a cloud of $p=\lfloor cf\rfloor$ vertices. Each cloud induces a clique with edge weight $\delta$, and all $p^2$ edges between the clouds have weight $w(e')$.}
    \label{fig:edgechange}
    \end{figure}

    Let $Q$ consist of all internal clique edges and, for each $\{u,v\} \in T'$, the two perfect matchings
    $$\{\{(u,i),(v,i)\}:i\in[p]\}
    \qquad\text{and}\qquad
    \{\{(u,i),(v,i+1)\}:i\in[p]\},$$
    with indices taken modulo $p$.  Their union is a cycle of length $2p$.  Every nontrivial cut of $Q$ has at least $p+1$ edges: if a cloud $C_u$ meets both sides, its clique contributes at least $p-1$ cut edges, and the cycle to any neighboring cloud in $T'$ contributes at least two more.  Otherwise, an edge of $T'$ crosses the corresponding partition of clouds, and both its matchings contribute $2p$ cut edges.  Thus $Q$ remains connected after any $\lfloor cf\rfloor=p$ edge failures, and is a $cf$-EFT connectivity preserver of $G$ with
    $$w(Q) \leq (2p+1)w(T').$$

    Let $H$ be an $f$-EFT $k$-spanner of $G$.  For each $e'=\{u,v\}\in E'$, $H$ must contain at least $f+1$ edges between $C_u$ and $C_v$.  Otherwise, fail precisely the at most $f$ such edges in $H$.  Since $p^2>f$, a direct input copy survives, with endpoint distance at most $w(e')$ in the surviving input graph.  Any path between its endpoints in the surviving spanner projects, after removing within-cloud steps, to a $u$-to-$v$ walk in $G'$ avoiding $e'$ and of no greater weight.  Erasing closed excursions gives a path avoiding $e'$, whose weight is greater than $k w(e')$ by the weighted girth of $G'$.  This contradicts stretch $k$.

    Therefore $w(H)\geq(f+1)w(G')$, and
    \begin{align*}
        \ell_{cf}(H \mid G) &\geq \frac{(f+1)w(G')}{(2p+1)w(T')}\\
        &\geq \frac{1}{2c}\ell(G')\\
        &= \Omega\left(\lambda\left(\frac{n}{cf},k+1\right)\right),
    \end{align*}
    where the last step uses $n/(cf)=n'p/(cf)\leq n'$ and monotonicity of $\lambda$.  Taking arbitrarily large $n'$ gives the required infinite family.
\end{proof}

\fi

\section{Future Directions: Light Vertex Fault Tolerant Spanners \label{sec:vft}}
In addition to the obvious open questions of a) determining the precise bound on the achievable $2f$-competitive lightness, and b) providing simultaneous bounds on both the lightness and the sparsity, an interesting open problem left by this paper is to determine the appropriate competition parameter for light \emph{vertex} fault tolerant (VFT) spanners. 
We can define vertex competitive lightness analogously to Definition \ref{def:complight}, but with respect to the set of $f$-VFT connectivity preservers rather than $f$-EFT connectivity preservers.
We then ask:

\begin{question}
What is the smallest function $\mathcal{C}(f)$ such that every $n$-node weighted graph $G$ has an $f$-VFT $k$-spanner with \textbf{vertex}-competitive lightness\footnote{It would also be interesting if, instead of a $\lambda$ reduction, one could match the state-of-the art non-faulty size/stretch tradeoff up to a $\texttt{poly}(f)$ factor.  That is, the goal would be to obtain stretch $(1+\eps)(2k-1)$ and lightness $O\left(\texttt{poly}(f, \eps^{-1}) \cdot n^{1+1/k}\right).$}
$$\ell_{\mathcal{C}(f)}(H \ \mid \ G) \le O\left( \texttt{poly}(f) \cdot \lambda(n, k+1) \right)?$$
\end{question}

Theorem \ref{thm:intromain} proves that the right answer in the edge-fault setting is exactly $\mathcal{C}(f) = 2f$, and it is not hard to see that the lower bound extends to show that the answer is $\mathcal{C}(f) \ge 2f$ for vertex faults.  But it is unclear whether a matching upper bound is possible.
Below, we discuss a few leads on this problem, as well as the various technical difficulties in bringing them to fruition.

\paragraph{Connected Dominating Sets.} 
Censor-Hillel, Ghaffari, and Kuhn \cite{CGK14} proposed that \emph{connected dominating sets} are the natural analog of spanning trees when considering vertex connectivity rather than edge connectivity, and they proved the following analog of the Nash-Williams theorem:
\begin{theorem} [\cite{CGK14}] \label{thm:CDS}
Every $f$-vertex-connected $n$-node graph contains a collection of $\Omega(f / \log n)$ vertex-disjoint CDSes.
\end{theorem}

So a natural approach to getting vertex fault tolerant light spanners is to replace the use of spanning trees and the Nash-Williams Theorem in our EFT result with CDSes and Theorem~\ref{thm:CDS}.  While there are some technical issues, this approach basically works, but with two very large caveats.  First, we only get $\Theta(f \log n)$-competitive lightness, rather than $2f$-competitive lightness.  Second, recall that our EFT approach actually needed an \emph{extension} of the Nash-Williams Theorem due to Chekuri and Shepherd~\cite{chekuri2009approximate} which allowed for the underlying graph to be less than $f$ connected.  If we use Theorem~\ref{thm:CDS} directly, rather than proving a similar extension, we would require that the underlying graph be at least $\Omega(f \log n)$-vertex connected.  Since we do not even want to assume $f$-connectivity, much less $\Omega(f \log n)$-connectivity, this is an important caveat.  However, while it is possible that one could prove an extension of Theorem~\ref{thm:CDS} which does not require high underlying connectivity, even then we would be limited to $\Omega(f \log n)$-competitive lightness.

\paragraph{Independent Spanning Trees.}
To overcome the fact that the CDS approach can only give a fairly large competition parameter of $\Theta(f \log n)$, we must consider other objects that guarantee many paths that do not overlap too much. 
Another object that would suffice are \textit{independent spanning trees}, defined as follows:
\begin{definition}
    Let $r$ be a vertex of graph $G = (V,E)$. A collection of spanning trees $T_1, \dots, T_k$ of $G$ are \textit{independent with root $r$} if for each vertex $v \in V$, the paths from $v$ to $r$ in $T_1, \dots, T_k$ are pairwise vertex disjoint.
\end{definition}

It was conjectured by~\cite{IR88} that every $k$-vertex connected graph contains a collection of $k$ independent spanning trees for every root $r$.  They proved this conjecture for $k=2$, and it has subsequently been proved for $k=3$~\cite{SS19}, $k=4$~\cite{HT18}, and $k=5$~\cite{AL23}, but it is still open for all $k \geq 6$.  While it is not nearly as simple as the CDS case (mostly due to the fact that the choice of root implies that paths actually \emph{do} overlap at the root), it turns out that one can similarly show that these objects can be used inside of our framework.  However, this suffers again from two major drawbacks.  First, this is still a conjecture.  It could very well be false, and even if true, it has been open for over 35 years with very little progress.  Second, it would still require that the underlying graph be $2f$-vertex connected; it would provide no guarantees for the more general case (which we can do in the EFT setting). 

\bibliographystyle{alpha}
\bibliography{refs}

\end{document}